\documentclass[a4paper,11pt]{article}

\pdfoutput=1 

\usepackage{jheppub, amsmath, amssymb,amsthm,dsfont,graphicx,mathtools,multirow,rotating,placeins,verbatim} 

\usepackage{mathrsfs} 

\usepackage{xcolor}
\usepackage{slashed}
\usepackage{tensor}
\usepackage[T1]{fontenc} 
\usepackage{float}

\usepackage{tikz-cd} 
\usepackage{tikz} 
\usepackage{todonotes} 
\usepackage{cases} 
\usepackage[utf8]{inputenc} 
\usepackage{hyperref} 
\usepackage{cleveref} 
\usepackage{mathrsfs} 
\usepackage{cancel} 
\usepackage[percent]{overpic} 
\usepackage{longtable} 
\usepackage{array} 
\usepackage{booktabs} 
\usepackage{tabularx} 
\usepackage{colortbl} 
\usepackage{xcolor} 
\usepackage{hhline} 
\usepackage{float} 
\usepackage{amsmath} 
\usepackage{amsthm} 
\usepackage{empheq} 
\usepackage{bbm} 
\usepackage{amsfonts} 
\usepackage{calligra} 
\usepackage{braket} 
\usepackage{enumitem} 
\usepackage[english]{babel} 
\usepackage{xcolor} 
\usepackage{stmaryrd} 
\usepackage{csquotes} 
\usepackage{subcaption} 

\newtheorem{theorem}{Theorem}[section] 
\newtheorem{definition}[theorem]{Definition} 
\newtheorem{lemma}[theorem]{Lemma} 
\newtheorem{remark}[theorem]{Remark} 
\newtheorem{proposition}[theorem]{Proposition} 
\crefname{subsubsection}{section}{sections} 
\crefname{lemma}{lemma}{lemmas} 

\DeclareMathAlphabet\mathbfcal{OMS}{cmsy}{b}{n}

\newcommand{\dd}{\mathrm{d}} 

\newcommand{\black}{\color{black}}

\newcommand{\zb}{\bar{z}}

\def\be{\begin{equation}}
\def\ee{\end{equation}}
\def\bi{\begin{itemize}}
\def\ei{\end{itemize}}
\newcommand{\beq}{\begin{eqnarray}}
\newcommand{\eeq}{\end{eqnarray}}

\setuptodonotes{linecolor=black!50,backgroundcolor=yellow!80,size=footnotesize}

\usepackage{leftindex}

\def\rf{\mathfrak{r}}
\def\zb{\bar{z}}

\newcommand{\C}{\mathbb{C}}

\newcommand{\Ib}{\mathcal{I}}
\newcommand{\tr}{\text{tr}}
\newcommand{\ad}{\text{ad}}

\newcommand{\ra}{\vartriangleleft} 
\newcommand{\la}{\vartriangleright} 
\newcommand{\blla}{\blacktriangleright} 

\title{Boundary Actions and Loop Groups: A Geometric Picture of Gauge Symmetries at Null Infinity}

\author[a,1]{Silvia Nagy,\note{Corresponding author.}}
\author[b,c]{Javier Peraza,}
\author[a]{Giorgio Pizzolo}

\affiliation[a]{Department of Mathematical Sciences, Durham University, Durham, DH1 3LE, UK}
\affiliation[b]{Perimeter Institute for Theoretical Physics, 31 Caroline St. N., Waterloo ON, Canada, N2L 2Y5}
\affiliation[c]{Centro de Matemática, Universidad de la República, Montevideo, Uruguay}

\emailAdd{silvia.nagy@durham.ac.uk}
\emailAdd{jperaza@pitp.ca}
\emailAdd{giorgio.pizzolo@durham.ac.uk}

\abstract{
In previous work \cite{Nagy:2024dme,Nagy:2024jua}, we proposed an extended phase space structure at null infinity accommodating large gauge symmetries for sub$^n$-leading soft theorems in Yang-Mills, via dressing fields arising in the Stueckelberg procedure. Here, we give an explicit boundary action controlling the dynamics of these fields. This allows for a derivation from first principles of the associated charges, together with an explicit renormalization procedure when taking the limit to null and spatial infinity, matching with charges proposed in previous work. Using the language of fibre bundles, we relate the existence of Stueckelberg fields to the notion of extension/reduction of the structure group of a principal bundle, thereby deriving their transformation rules in a natural way, thus realising them as Goldstone-like objects. Finally, this allows us to give a geometric picture of the gauge transformation structure at the boundary, via a loop group coming from  formal expansions in the coordinate transversal to the boundary.
}

\begin{document} 

\maketitle


\section{Introduction}

In the search for a flat space holographic principle, symmetries at the boundary of spacetime have played a foundational role. Their ability to supply a first-principles derivation of soft theorems in scattering amplitudes has sparked extensive subsequent work\footnote{See reviews \cite{Strominger:2017zoo,Pasterski:2021raf,Raclariu:2021zjz} and references within.}. Soft theorems are factorisations of amplitudes in the limit when the energy of one of the particles vanishes. More generally, they can be written as a series expansion in the small energy, thus capturing further information from the subleading orders\footnote{See \cite{Campiglia2015,Campiglia2020,Campiglia:2021oqz,Strominger2014a,Lysov:2014csa,Donnelly:2016auv,Speranza:2017gxd,Freidel:2020ayo,Freidel:2020svx,Freidel:2020xyx,Freidel:2021dxw,Ciambelli:2021nmv,Freidel:2021cjp,Freidel:2021dfs,Ciambelli:2021vnn,Geiller:2024bgf,Geiller:2022vto,Mao:2017tey,Bern:2014vva,Saha:2019tub,Campiglia:2024uqq,Upadhyay:2025ged,Fuentealba:2025ekj,Briceno:2025ivl} for relevant work.}. Unfortunately, it is not immediately obvious how to capture subleading terms from symmetries acting canonically on the phase space at the boundary.

In \cite{Nagy:2024dme} and \cite{Nagy:2024jua}, we developed a systematic framework for extending the phase space so as to accommodate symmetries associated with all sub$^n$-leading orders in Yang–Mills theory. This construction builds on earlier results at the first subleading order \cite{Campiglia2020,Campiglia:2021oqz} and in the simplified set-ups of the self-dual sector \cite{Nagy:2022xxs} and QED \cite{Peraza:2023ivy}. We further gave an interpretation of the new ingredients required for the extended phase space as Stueckelberg fields. These are Goldstone-like fields on which the symmetries are realised non-linearly. The corresponding charges were shown to generate algebras that admit consistent truncations at any sub$^n$-leading order. Furthermore, in the self-dual subsector, these structures reproduce the infinite-dimensional  S-algebra previously studied in \cite{Strominger:2021mtt,Freidel:2023gue,Hu:2023geb}.

In this work, we take an important next step by formulating an action principle for the Stueckelberg fields. Remarkably, their dynamics is governed entirely by an action localised on the boundary, from which we are able to recover the charges previously derived in \cite{Nagy:2024dme,Nagy:2024jua}. This construction thus brings us closer to a realisation of a holographic principle and also provides a framework that can be extended to more general settings. Additionally, we provide an explicit renormalisation procedure -- both in the radial coordinate $r$ and in the retarded time $u$ --  which was absent in our earlier analysis, building on previous work in \cite{Freidel:2019ohg,Compere:2020lrt, Ruzziconi:2020wrb, Chandrasekaran:2021vyu, Peraza:2023ivy, Campiglia:2018dyi, Romoli:2024hlc, Riello:2024uvs}.

Using the language of fibre bundles, and building on \cite{Attard_2018,Fran_ois_2021,Popov:2021gsa,Popov:2022ytx}, we present a formal account of the Stueckelberg field, deriving systematically both its existence and its transformation properties. This framework, in turn, provides a clear geometric picture of the structure of gauge transformations at the boundary (null infinity $\Ib^+$ in our case). By considering a formal expansion on the transversal parameter to the boundary, we construct a new structure group, which takes the form of a loop group. Its associated gauge group contains transformations acting on the Stueckelberg field, while the gauge vector remains invariant under these transformations. 

The article is structured as follows. In \autoref{sec:transformation_at_group_level}, we review the Stueckelberg procedure and extend our earlier results in two directions: by deriving the transformations of the fields at the group level, not just the algebra level, and by extending the procedure such that all the symmetry transformations -- including the leading order one -- are realised on the Stueckelberg fields. This extension is essential for the construction of the boundary action encoding the dynamics of the Stueckelberg fields in \autoref{sec: Lagrangian}. In \autoref{sec: renormalisation}, we implement a procedure for renormalising this action by exploiting the freedom in the pre-symplectic potential. In \autoref{sec: Fibre bundle}, we use the framework of principal fibre bundles in order to provide a geometrical setting in which we can derive extensions and reductions of gauge symmetries. We then use these results to describe gauge transformations at the boundary as elements of gauge groups corresponding to suitable loop groups. Finally, we conclude in \autoref{sec: Conclusions}.

\section{The Stueckelberg procedure: review and extension}
\label{sec:transformation_at_group_level}

We are working in YM theory, with the standard action in the bulk given by 
\be \label{YM_orig_action}
 S=-\frac{1}{2}\int_M\tr(*\mathcal{F}\wedge \mathcal{F})
\ee 
where $M$ is an asymptotically flat spacetime with local coordinates $x^\mu=(\mathfrak{r},\mathbf{y})$. Here $\mathbf{y}$ parametrise the boundary, in particular in Bondi coordinates we would have $\mathfrak{r}=r$ and $\mathbf{y}=(u,z,\bar{z})$. The gauge field admits a general expansion
\be\label{A log expansion} 
\mathcal{A}_\mu (x)=\sum_{n,k}\frac{\text{log}^k \mathfrak{r}}{\mathfrak{r}^n} A_\mu^{(-n;k)}(\mathbf{y}) \ ,
\ee 
with $n$ and $k$ such that $\lim_{\mathfrak{r}\to\infty}\frac{\text{log}^k \mathfrak{r}}{\mathfrak{r}^n}=\mathcal{O}(1)$. This is the natural physical fall-off expected for the gauge field. The gauge parameters which preserve this fall-off will be of the form\footnote{We set all ``small'' gauge transformations (i.e. those which vanish at the boundary) to 0, as they don't contribute to the soft theorems.}
\be\label{Lambda0def}
\Lambda_0=\Lambda_0(\mathbf{y}) \;.
\ee 
However, in order to access sub$^n$-leading soft effects in the scattering amplitudes via the Ward identity, we need, in general, a large gauge parameter such that
\begin{equation}
    \label{def: Lambda plus}
    \Lambda_+(x)\coloneqq\sum_{n,k}\mathfrak{r}^n\log^k\mathfrak{r}\,\Lambda_+^{(n;k)}(\mathbf{y})\,, \quad \text{with}\quad n,k\ \ \text{such that} \ \ \lim_{\mathfrak{r}\rightarrow\infty}\mathfrak{r}^n\log^k\mathfrak{r}=\infty\;.
\end{equation}
These will of course violate the fall-off of the gauge field, and the goal of the Stueckelberg procedure is to introduce fields on which these act naturally.

\subsection{Group theoretic formulation of the Stueckelberg procedure}
\label{subsec:review Stueck}

Here we review the procedure we developed in \cite{Nagy:2024dme,Nagy:2024jua}, which introduces the Stueckelberg fields capable of accommodating local transformations with overleading parameters $\Lambda_+$ given in \eqref{def: Lambda plus}. Moreover, we extend the procedure to obtain the full group transformations of these fields, from which the transformations in \cite{Nagy:2024dme,Nagy:2024jua} are recovered as the linear approximation in the parameter.

To start, let $\mathcal{H},\mathcal{G}$ be the gauge groups generated, respectively, by the parameters
\begin{equation}
    \label{def:parameter Lambda}
    \Lambda_0(\mathbf{y}) \quad \text{and} \quad \Lambda(x)\coloneqq\Lambda_0(\mathbf{y})+\Lambda_+(x)\;,
\end{equation}
with $\Lambda_0(\mathbf{y})$ and $\Lambda_+(x)$ given in \eqref{Lambda0def} and \eqref{def: Lambda plus}, respectively.

In particular, this implies $\mathcal{H}\subset\mathcal{G}$. The gauge field $\mathcal A$ on $M$ naturally transforms only under $\mathcal{H}$, that is to say
\begin{equation}
    \label{eq:initial YM field, case 1}
   \mathcal A^\eta=\eta^{-1}(\mathcal A+\dd)\eta \,, \quad \eta\coloneqq e^{-i\Lambda_0(\mathbf{y})}\in\mathcal{H} \;.
\end{equation}
Starting from $\mathcal{A}$, we want to define a new gauge field $\tilde{\mathcal{A}}$ that transforms under  $\mathcal{G}$, i.e.
\begin{equation}
    \label{eq:transformation A tilde}
    \tilde{\mathcal{A}}^\gamma=\gamma^{-1}(\tilde{ \mathcal A}+\dd)\gamma \,, \quad \gamma\coloneqq e^{-i\Lambda(x)}\in\mathcal{G} \;.
\end{equation}
To do so, we promote the parameter $\Lambda_+$ to a field $\Psi_+$\footnote{In this paper we denote by $\Psi_+$ the Stueckelberg field introduced in \cite{Nagy:2024dme,Nagy:2024jua}, where it was simply written as $\Psi$. The change in notation is motivated by the fact that here we introduce two distinct dressing fields, cf. equations \eqref{def:dressing field s plus} and \eqref{fullStueckField}, transforming under different symmetry groups. The subscript \enquote{$+$} allows us to clearly distinguish between them. In \cite{Nagy:2024dme,Nagy:2024jua}, such a distinction was unnecessary, hence the simpler notation without the subscript was adopted for brevity. Moreover, in contrast with \cite{Nagy:2024dme,Nagy:2024jua}, here we will not consider field-dependent gauge transformations.} which we use to dress the initial gauge field as follows:
\begin{equation}
    \label{def:A tilde s plus}
    \tilde {\mathcal A}\coloneqq s_+^{-1}(\mathcal A+\dd) s_+ \;,
\end{equation}
where
\begin{equation}
    \label{def:dressing field s plus}
    s_+(x)\coloneqq e^{-i\Psi_+(x)} \,, \quad \Psi_+(x)\coloneqq\sum_{n,k}\mathfrak{r}^n\log^k\mathfrak{r}\,\Psi_+^{(n;k)}(\mathbf{y}) \;,
\end{equation}
and $n,k$ are as in \eqref{def: Lambda plus}. We will refer to both $s_+$ and $\Psi_+$ as \textit{dressing} fields or, equivalently, \textit{Stueckelberg} fields\footnote{In \cref{sec: Fibre bundle}, in order to align with the community studying Dressing Field Methods \cite{Attard_2018,Fran_ois_2021,Fournel:2012cr,Francois:2015oca,Attard:2016stx,Francois:2017akk,Mathieu:2019lgi,Berghofer:2021ufy,FrancoisAndre:2023jmj,Berghofer:2024plf,Francois:2025sic}, we adopt a slightly different terminology. In the rest of the paper, however, the two terms may be used interchangeably.}. Note that although \eqref{def:A tilde s plus} has the form of a gauge transformation, it is conceptually different, as $\Psi_+(x)$ is a field, and not a parameter. This will be further clarified in \autoref{sec: Fibre bundle}.    

We now determine the correct transformation rule of $\Psi_+$ under $\mathcal{G}$ such that the dressed field $\tilde{A}$ transforms as required, i.e. as a YM field under the gauge group with parameter $\Lambda$. To begin with, observe that for all $\gamma\in \mathcal{G}$ there exists a unique pair
\be
\label{eq:gamma decomposition}
(\nu_\gamma,\eta_\gamma)\in \mathcal{N}\times\mathcal{H},\quad\text{such that}\quad\gamma=\nu_\gamma\eta_\gamma\ ,
\ee 
where $\mathcal{N}$ is the group with elements of the form $e^{-i\Lambda_+}$. Indeed,
\begin{equation} \label{eq:normality}
    \gamma=e^{-i\Lambda}=e^{-i\Lambda_0-i\Lambda_+}=\overbrace{\left(\prod_{n=\infty}^2e^{-iC_n(\Lambda_0,\Lambda_+)}\right)e^{-i\Lambda_+}}^{\in\mathcal{N}}\overbrace{e^{-i\Lambda_0}}^{\in\mathcal{H}} \;,
\end{equation}
where $C_n(\Lambda_0,\Lambda_+)$ is a homogeneous expression of $n-1$ nested commutators of $\Lambda_0$ and $\Lambda_+$, and as such it has the same expansion as $\Lambda_+$. The last equality above is known as the left-oriented Zassenhaus formula, which can be regarded as the inverse of the Baker–Campbell– Hausdorff (BCH) formula, see e.g. \cite{Casas:2012nqk} for a reference. One finds that the required transformation for the dressing field $s_+$ under $\mathcal{G}$ is
\begin{equation}
    \label{res:trasformation of s plus}
    s_+^\gamma=\eta_\gamma^{-1}s_+\gamma \;,
\end{equation}
where $\eta_\gamma$ is as in \eqref{eq:gamma decomposition}. We will give a formal derivation of \eqref{res:trasformation of s plus} in \autoref{sec: Fibre bundle}. For now we verify directly that it gives the desired transformation for $\tilde{\mathcal A}$:

\begin{equation}\label{tildeAgammaDerivation}
    \begin{aligned}
        \tilde{\mathcal A}^\gamma&=(s_+^\gamma)^{-1}\mathcal A^\gamma s_+^\gamma+(s_+^\gamma)^{-1}\dd s_+^\gamma \\
        &=\gamma^{-1}s_+^{-1}\eta_\gamma(\eta_\gamma^{-1}\mathcal A\eta_\gamma+\eta_\gamma^{-1}\dd\eta_\gamma)\eta_\gamma^{-1}s_+\gamma+\gamma^{-1}s_+^{-1}\eta_\gamma\dd(\eta_\gamma^{-1}s_+\gamma) \\
        &=\gamma^{-1}(s_+^{-1}\mathcal As_++s_+^{-1}\dd s_+)\gamma+\gamma^{-1}\dd\gamma \\
        &=\gamma^{-1}\tilde{\mathcal A}\gamma+\gamma^{-1}\dd\gamma \;,
    \end{aligned}
\end{equation}
where we used the fact that $\mathcal A$ transforms only under $\Lambda_0$, that is to say $\mathcal A^\gamma=\mathcal A^{\eta_\gamma}$. To find the variation of $\Psi_+$ corresponding to the transformation \eqref{res:trasformation of s plus} we write
\begin{equation}
    \begin{cases}
        \delta s_+=\delta e^{-i\Psi_+}=e^{-i\Psi_+}\mathcal{O}_{-i\Psi_+}(-i\delta\Psi_+) \\
        \delta s_+=\eta_\gamma^{-1}s_+\gamma-s_+=e^{i\Lambda_0}e^{-i\Psi_+}e^{-i\Lambda}-e^{-i\Psi_+} \;,
    \end{cases}
\end{equation}
where $\mathcal{O}$ is the operator
\begin{equation} \label{eq:curly_O_def}
    \mathcal{O}_X\coloneqq\frac{1-e^{-\ad_X}}{\ad_X}\; ,
\end{equation}
to be understood as a formal series expansion in $\ad_X(Y)=[X,Y]$, see \cite{Nagy:2024dme,Nagy:2024jua} for details. Thus, we find
\begin{equation}
    \delta\Psi_+=i\mathcal{O}_{-i\Psi_+}^{-1}(e^{i\Psi_+}e^{i\Lambda_0}e^{-i\Psi_+}e^{-i\Lambda}-1)\;.
\end{equation}
At linear order in $\Lambda_0$ and $\Lambda_+$:
\begin{equation}
    \begin{aligned}
        \delta\Psi_+&=i\mathcal{O}_{-i\Psi_+}^{-1}[e^{i\Psi_+}(1+i\Lambda_0)e^{-i\Psi_+}(1-i\Lambda)-1] \\
        &=\mathcal{O}_{-i\Psi_+}^{-1}(\Lambda-e^{i\Psi_+}\Lambda_0e^{-i\Psi_+}) \\
        &=\Lambda_+ - \frac{i}{2}[\Psi_+,\Lambda+\Lambda_0]+\dots
    \end{aligned}
\end{equation}
This transformation law for $\Psi_+$ coincides exactly with the one previously obtained in \cite{Nagy:2024dme,Nagy:2024jua}, with the above derivation having the advantage of being more general and formulated at the group level. In particular, as discussed in \cite{Nagy:2024dme,Nagy:2024jua}, at $0$-th order in the field  $\Psi_+$ transforms via a shift, thus it can be interpreted as the Goldstone boson associated with the spontaneous breaking of the extended gauge symmetry in the bulk. We refer the reader to \cite{Nagy:2024jua} for a more thorough discussion.

\subsection{Decoupling all the symmetries via an extended Stueckelberg procedure} \label{sec:two_step_stueck}
In the previous subsection we started from the gauge field $\mathcal A$ transforming under the group $\mathcal{H}$ generated by leading order gauge transformations $\Lambda_0(\mathbf{y})$, and defined a new gauge field $\tilde{\mathcal A}$ transforming under the group $\mathcal{G}\supset\mathcal{H}$, generated by leading and overleading transformations, see \eqref{eq:initial YM field, case 1} and \eqref{eq:transformation A tilde}. This was done by promoting the parameter $\Lambda_+$ to a field $\Psi_+$ and introducing the dressing field $s_+= e^{-i\Psi_+}$. As we will see in \autoref{sec: Lagrangian}, for the purpose of constructing the boundary Lagrangian it is actually desirable to have an invariant gauge field, such that its transformation has been \enquote{transferred} onto a new field, living on the boundary. We will achieve this here by another application of the Stueckelberg procedure, and identify the required boundary field with a new Stueckelberg field.

Given that $\mathcal A$ transforms only under $\Lambda_0(\mathbf{y})$, we can obtain a fully invariant gauge field $\mathbfcal A$ via the dressing 
\be \label{A0dressing}
 \mathbfcal A\coloneqq s_0^{-1}(\mathcal A+\dd) s_0 
 \ ,
\ee 
where
\be
s_0=s_0(\mathbf{y})\equiv e^{-i\Psi_0(\mathbf{y})}\ .
\ee 
Then, if $s_0$ transforms as\footnote{See \autoref{sec: Fibre bundle} for a formal derivation via fibre bundles.}
\be \label{s_0Transform}
s_0^\eta=\eta^{-1} s_0\; ,
\ee 
with $\eta$ as in \eqref{eq:initial YM field, case 1}, it is easy to show, via a calculation similar to that in \eqref{tildeAgammaDerivation}, that $\mathbfcal A$ is invariant under the full group $\mathcal G$, i.e. 
\be \label{A0invariance}
\mathbfcal A^\gamma=\mathbfcal A^{\eta_\gamma}=\mathbfcal A\;.
\ee 
It is now straightforward to write $\tilde{\mathcal{A}}$ in terms of $\mathbfcal A$ by solving for $\mathcal{A}$ in \eqref{A0dressing} and plugging this into \eqref{def:A tilde s plus} to get
\be 
\label{tildeAinTermsOfA0}
\tilde{\mathcal A}=s_+^{-1}s_0 \mathbfcal A s_0^{-1} s_+
+s_+^{-1}s_0\dd\left( s_0^{-1}s_+\right) \; .
\ee
At this stage, recalling that $s_+=e^{-i\Psi_+(x)}$, we can define
\be 
\label{fullStueckField}
s=s_0^{-1}s_+=e^{i\Psi_0(\mathbf{y})}e^{-i\Psi_+(x)}=e^{i\Psi_0(\mathbf{y})-i\tilde\Psi_+(x)} 
\equiv e^{-i\Psi(x)} \; ,
\ee 
where $\tilde\Psi_+(x)$ have the same fall-off as $\Psi_+(x)$ in \eqref{def:dressing field s plus}, as they are obtained as nested commutators of the latter with $\Psi_0(\mathbf{y})$, via the BCH formula.

We remark that the 2-stage Stueckelberg procedure consisting of 
\begin{itemize}
\item dressing $\mathcal A$ to achieve the extension of the symmetry grup from  $\mathcal{H}$ to  $\mathcal{G}$, as described in \autoref{subsec:review Stueck}.
\item introducing the invariant $\mathbfcal A$ by transferring its transformation to the new Stueckelberg field $s_0$, as described above in \eqref{A0dressing}-\eqref{A0invariance}
\end{itemize} 
can be reinterpreted, via \eqref{tildeAinTermsOfA0} and \eqref{fullStueckField}, as a single step Stueckelberg procedure, 
\begin{equation}\label{single step Stueck}
    \tilde {\mathcal A}\coloneqq s^{-1}(\mathbfcal A +\dd) s \;,
\end{equation}
with $s$ as in \eqref{fullStueckField}, which extends an invariant field $\mathbfcal A$ to a field $\tilde {\mathcal A}$ transforming under $\mathcal G$. We can derive the transformation of $s$ from \eqref{s_0Transform} and \eqref{res:trasformation of s plus}
\be\label{full s transform} 
s^\gamma =\left(s_0^{-1}\right)^\gamma\left(s_+\right)^\gamma=s_0^{-1}\eta_\gamma \eta_\gamma^{-1}s_+\gamma=s\gamma 
\ee 
There is an alternative way to derive the transformation rule for $s$ in \eqref{full s transform}. We observe that the single step Stueckelberg procedure in \eqref{single step Stueck} can actually be thought of as a special case of the procedure described in \autoref{subsec:review Stueck}, with the group $\mathcal{H}$ reduced to the identity element. Indeed, saying that the initial $\mathbfcal A$ does not transform is equivalent to saying that it transforms only under the identity, that is
\begin{equation}
    \mathbfcal A^\epsilon=\epsilon^{-1}(\mathbfcal A+\dd)\epsilon=\mathbfcal A \,, \quad \epsilon\coloneqq \mathbf{1}_{\mathcal{G}} \;.
\end{equation}
This automatically extends $s_+$ to $s$ when performing the Stueckelberg procedure, and the transformation rule for $s$ is then just a special case of \eqref{res:trasformation of s plus} with $\eta_\gamma=\mathbf{1}_{\mathcal{G}}$, so we have 
\be 
s^\gamma=\mathbf{1}_{\mathcal{G}} s\gamma = s\gamma  \; .
\ee 
in agreement with \eqref{full s transform}. Finally, the corresponding transformation rule for the Stueckelberg field $\Psi$ introduced in \eqref{fullStueckField} is given by
\begin{equation}
    \delta\Psi=i\mathcal{O}_{-i\Psi}^{-1}(e^{-i\Lambda}-1)\;,
\end{equation}
where $\mathcal{O}$ is defined in \eqref{eq:curly_O_def}. At linear order in $\Lambda$ the above equation reads
\begin{equation}
    \begin{aligned}
        \delta\Psi&=\mathcal{O}_{-i\Psi}^{-1}(\Lambda) \;.
    \end{aligned}
\end{equation}

\section{Edge mode Lagrangian for YM}
\label{sec: Lagrangian}

We have seen that the dressing procedure in the previous section allows for the introduction of Stueckelberg fields on which overleading symmetries, related to sub$^n$-leading soft theorems, can be realised naturally. The next logical step is to construct a Lagrangian governing the dynamics of these fields. This would allow for the construction of a symplectic form from first principles, starting from a Lagrangian formulation that contains the Stueckelberg field from within, and not as an \emph{ad hoc} input of the theory.

Our approach for constructing the action here takes inspiration from \cite{Popov:2021gsa,Geiller:2019bti,Blommaert:2018oue}, and earlier works in \cite{GERVAIS197655,Donnelly:2016auv, Harlow:2019yfa} \footnote{A recent work, \cite{Campoleoni:2025bhn}, constructs finite actions for Maxwell theory and linearized gravity, by introducing boundary counter-terms in the action. It would be interesting to match their additional terms to the ones we use in this work.}. Indeed, we will see that in the leading order limit, our action will reduce to that in \cite{Geiller:2019bti,Blommaert:2018oue}. We propose the following action: 
\be  
 S=S_M+S_{\partial M}
 =-\frac{1}{2}\int_M\tr(*\tilde{\mathcal F}\wedge \tilde{\mathcal F})+\int_{\partial M}\tr\left[\tilde j\wedge\tilde{\mathcal A}\right]^{\text{ren}}
\ee
where the boundary action is appropriately renormalised to avoid possible divergences. One could take the more standard approach of renormalising the pre-symplectic potential $\theta$, when computing the charge via the phase space formalism. Indeed, we will demonstrate how to rigorously renormalise $\theta$ in \autoref{sec: renormalisation}, and verify that it gives the same charge as the one we will compute directly from the action in this section. $\tilde{\mathcal A}$ is as given in \eqref{single step Stueck}, copied below for convenience:
\be 
\tilde {\mathcal A}\coloneqq s^{-1}(\mathbfcal A +\dd) s \;.
\ee 
Consequently, 
\be\label{tildeFboldF} 
\tilde{\mathcal F}\coloneqq s^{-1}\mathbfcal Fs \, , 
\ee
with $\mathbfcal F$ the field strength of $\mathbfcal A$. We postulate that
\be 
\tilde{j}=s^{-1}\mathbf js\;.
\ee 
This will be justified a posteriori, as we will see shortly that the source two-form $\mathbf j$ will become identified with the dual of the field strength. The Lagrangian then becomes:
\be \label{bound Lag}
S=S_M+S_{\partial M}=-\frac{1}{2}\int_M\tr(* \mathbfcal F\wedge \mathbfcal F)+\int_{\partial M} \tr[\mathbf j\wedge(\mathbfcal A+\dd ss^{-1}) ]^\mathrm{ren}\;.
\ee
Note that the bulk action is identical to \eqref{YM_orig_action}, by virtue of the definition of $\mathbfcal A$, see \eqref{A0dressing}. The action is invariant under the gauge group $\mathcal{G}$, under which the constituent fields transform as\footnote{See \eqref{A0invariance} and \eqref{full s transform}. The transformation law for $\mathbf{j}$ is an ansatz at this stage, which will be verified later, see \eqref{jF identification}.} 
\be \label{bolded_transforms}
\mathbfcal A^\gamma=\mathbfcal{A}\,,\quad\mathbfcal F^\gamma=\mathbfcal F\,,\quad s^\gamma=s\gamma\,,\quad
\mathbf j^\gamma=\mathbf j \; ,
\ee
when $\mathbf{j}$ obeys the conservation law 
\be \label{cons j}
  \dd \mathbf j-[\dd ss^{-1},\mathbf{j}]= 0\qquad \Leftrightarrow\qquad
  s\dd(s^{-1} \mathbf j s)s^{-1}=0
\ee 
following from the e.o.m.\ from $s$, as we will see below. We remark that \eqref{cons j} is just the gauge covariant derivative of the source $\mathbf{j}$ with respect to the pure gauge field $\dd ss^{-1}$. 

In order to simplify the renormalisation procedure, we will assume a non-logarithmic fall-off for the gauge field, source and Stueckelberg fields: 
\begin{equation}\label{rexpforL}
    \mathbfcal A(x)=\sum_{n\leq0}\mathfrak{r}^n \mathbf A^{(n)}(\textbf{y})\; , \quad
    \mathbf j(x)=\sum_{n\leq0}\mathfrak{r}^n \mathbf j^{(n)}(\textbf{y})
    \; , \quad
    s(x)=\sum_{n\geq0}\mathfrak{r}^n s^{(n)}(\textbf{y})\;.
\end{equation}
Then we simply write the renormalised boundary action as 
\be 
\begin{aligned}
S_{\partial M}=&\int_{\partial M}\tr\left[\mathbf j\wedge(\mathbfcal A+\dd ss^{-1})\right]^{\text{ren}}\\
&=\int_{\partial M}\tr\left[\mathbf j\wedge(\mathbfcal A+\dd ss^{-1})\right]^{(0)}\\
&=\int_{\partial M}\tr\sum_{n=0}^{\infty}\mathbf j^{(-n)}\wedge(\mathbfcal A+\dd ss^{-1})^{(n)} \;,
\end{aligned}
\ee
meaning that the renormalisation in this case consists of just taking the 0th order in the $\mathfrak{r}$ expansion. Note that we can think of the objects entering $S_{\partial M}$ above as a collection of purely boundary fields, in particular for the source and the Stueckelberg field we will have two infinite sets of boundary fields, labelled by $n$:
\be 
\mathbf j^{(-n)}(\textbf{y}), \quad
s^{(n)}(\textbf{y}),\quad 
n\geq 0 \;,
\ee 
and the $\mathfrak{r}$-expansion in \eqref{rexpforL} can be though of as just a useful book-keeping device. 
We can then make use of the results in \autoref{App:some derivations} to write the variation of the action
\begin{align}
    \delta S_M&=-\int_M\tr(\delta \mathbfcal A\wedge \dd_{ \mathbfcal A}*\mathbfcal F)-\int_{\partial M}\tr(\delta \mathbfcal A\wedge *\mathbfcal F)\label{var1} \\
    \delta S_{\partial M}
    &=\int_{\partial M}\tr[\mathbf j\wedge\delta \mathbfcal A -\left(s\dd\left(s^{-1} \mathbf j s\right)s^{-1}\right)\delta ss^{-1} +\dd(\mathbf j \delta ss^{-1})]^{(0)} \; .\label{var2}
\end{align}

The bulk equation of motion, coming from \emph{variations of the gauge field which vanish at the boundary}, is of course just the standard YM equation
\be
    \dd_{\mathbfcal A} * \mathbfcal F= 0 \; , 
\ee
In order to obtain the boundary equations, \emph{we formally allow the variation of $\mathbfcal A$ to have an expansion in both positive and 0 powers of $\mathfrak{r}$} so that the relevant boundary term is
\be 
\int_{\partial M}\tr\left[\delta \mathbfcal A\wedge \left(\mathbf{j}-*\mathbfcal F\right)\right]^{(0)}
=\int_{\partial M}\tr\sum_{n=0}^{\infty}\delta \mathbf A^{(n)}\wedge \left(\mathbf j-* \mathbf F\right)^{(-n)}\;,
\ee 
from which we read off:
\be \label{jF identification}       
\mathbf j^{(-n)}=* \mathbf F^{(-n)},\ n\geq 0\quad \Rightarrow\quad \mathbf{j}=*\mathbfcal F\;.
\ee
Finally, as previously advertised, the e.o.m. for $s$ gives the current conservation in \eqref{cons j}.

\subsection{Charges}
The symplectic potential can be read off from \eqref{var1} and \eqref{var2}\footnote{\label{footnote full theta} We remark that a slightly different form, namely
\be\label{full theta}
    \theta=-\tr\left(\delta \mathbfcal A\wedge *\mathbfcal F-\dd(\mathbf j\delta ss^{-1})\right)     
\ee
is given in \cite{Geiller:2019bti,Blommaert:2018oue,Donnelly:2016auv, Harlow:2019yfa}. This accommodates an additional symmetry with a left action on the fields
\be 
\delta^l_\alpha s = -\alpha s , \quad \delta^l_\alpha \mathbfcal A = \dd_{\mathbfcal A} \alpha, \quad \delta^l_\alpha \mathbf j = [\mathbf j,\alpha] \; ,
\ee
which is seen as a true gauge symmetry, as the charge computed from \eqref{full theta} vanishes.} 
\be
\begin{aligned} \label{eq:symplectic_potential}
    \theta&=- \tr\ \dd(\mathbf j\delta ss^{-1})^{(0)} \\
    &=-\tr\ \dd(*\mathbfcal F\delta ss^{-1})^{(0)}\;,
\end{aligned}    
\ee
where we made use of \eqref{jF identification} to get to the second line. Next we compute the symplectic form:
\begin{equation}
    \begin{aligned}
        \omega[\delta,\bar\delta]&\coloneqq\tr\left[\bar\delta\theta(\delta)-\delta\theta(\bar\delta)\right]^{(0)} \\
        &=\tr\left[-\bar\delta\dd(*\mathbfcal F\delta ss^{-1})-(\delta\leftrightarrow\bar\delta)\right]^{(0)} \\
        &=-\tr\left[\dd(\bar\delta(*\mathbfcal F)\delta ss^{-1}-*\mathbfcal F\delta s s^{-1}\bar\delta ss^{-1})-(\delta\leftrightarrow\bar\delta)\right]^{(0)} \;.
    \end{aligned}
\end{equation}
Choose $\delta$ to be the infinitesimal variation for the gauge transformation \eqref{bolded_transforms}:
\begin{equation}
    \begin{aligned}
        \delta_\Lambda s=-i s\Lambda\,,\quad \delta_\Lambda \mathbfcal A=0\,,\quad \delta_\Lambda \mathbfcal F=0 \; ,
    \end{aligned}
\end{equation}
then the symplectic form is
\be
    \begin{aligned}
        \omega&=-\tr\ \dd(\bar\delta(*\mathbfcal F)\delta_\Lambda ss^{-1}-*\mathbfcal F\delta_\Lambda s s^{-1}\bar\delta ss^{-1}+*\mathbfcal F\bar\delta s s^{-1}\delta_\Lambda ss^{-1})^{(0)} \\
        &=i\tr\ \dd(\bar\delta(*\mathbfcal F)s\Lambda s^{-1}-*\mathbfcal Fs\Lambda s^{-1}\bar\delta ss^{-1}+*\mathbfcal F\bar\delta s \Lambda s^{-1})^{(0)} \\
        &=i\tr\ \bar\delta\dd(*\mathbfcal Fs\Lambda s^{-1})^{(0)}\; .
    \end{aligned}
\ee
Thus finally the charge density is
\begin{equation} \label{charge_density}
    \tilde{q}_\Lambda=i\mathrm{tr}(\Lambda*s^{-1}\mathbfcal Fs)^{(0)} =i\mathrm{tr}(\Lambda*\tilde{\mathcal F})^{(0)} \; ,
\end{equation}
using \eqref{tildeFboldF}. This is exactly as computed in our previous work \cite{Nagy:2024dme,Nagy:2024jua}. Here the Lagrangian formulation allows for a first principles derivation of the charge, as well as a well-defined renormalisation procedure, which will be further detailed in \cref{sec: renormalisation}. 

We recover the usual charge algebra corresponding to Yang-Mills,

\be 
\{\tilde  q_{{\Lambda}_1},\tilde q_{{\Lambda}_2}\}=\tilde q_{-i[{\Lambda}_1,{\Lambda}_2]}\; .
\ee
The corresponding perturbative expression (both in $\Psi$ and $r$) is 
\be
\label{charge_algebra_up_to_n}
    \{\overset{n-k}{q}_{\Lambda_1^{(k)}} , \overset{n-j}{q}_{\Lambda_2^{(j)}} \} =\begin{cases}
        \overset{n-k-j}{q}_{- i[\Lambda_1^{(k)}, \Lambda_2^{(j)}]} \quad \text{if } j+k \leq n \\
        \quad 0 \qquad\qquad\qquad\;\; \text{otherwise}
    \end{cases}\; ,
\ee
where $\overset{j}{q}_{\Lambda^{(l)}}$ is the approximation up to $j$-th order in $\Psi$ of $\tilde{q}_{\Lambda^{(l)}}$. See \cite{Nagy:2024jua} for more details.



\black

\section{Renormalization}
\label{sec: renormalisation}

Having an explicit Lagrangian, in which both the bulk field $\mathbfcal{A}$ and the bundary field $s$ are present, allows us to renormalize the symplectic potential using a standard procedure developed in \cite{Freidel:2019ohg} for divergences near null infinity. Once we have a Lagrangian for the theory, which includes the degrees of freedom at the boundary, we can argue along the lines of \cite{Freidel:2019ohg} and use the equations of motion to derive a well-defined renormalized symplectic potential. Both sources of divergences, $r$- and $u$- expansions, can be treated in a ladder-like fashion. Applications of this procedure have been done in the abelian case (Maxwell) in \cite{Peraza:2023ivy} in terms of the limit $t \rightarrow +\infty$ and $u\rightarrow -\infty$, to recover the sub$^n$-soft theorem charges \cite{Campiglia:2018dyi}. It has also been done in \cite{Romoli:2024hlc} for the renormalization of asymptotic charges in the two-form formalism. Applications to gravity can be found in \cite{Campiglia:2016efb,  Ruzziconi:2020wrb, Chandrasekaran:2021vyu, Compere:2020lrt}, and in \cite{Riello:2024uvs} an exhaustive study was given in the context of renormalization of conformal infinity in gravity.

In this section, we consider the action in \eqref{bound Lag} without any renormalization,
\be  
S= -\frac{1}{2}\int_M\tr(* \mathbfcal F\wedge \mathbfcal F)+\int_{\partial M} \tr[\mathbf j\wedge(\mathbfcal A+\dd ss^{-1}) ]  \; ,
\ee
and show that the divergences in the pre-symplectic potential derived from it can be appropriately renormalized. Once we have a well-defined, non-divergent pre-symplectic potential, its boundary component can be shown to coincide with the one in \eqref{eq:symplectic_potential}, and thus the charges can be computed, giving the expression in \eqref{charge_density}.

For definiteness, we fix the space of functions and expansions that we are working on as explained in the previous section. We also take Bondi coordinates to simplify the treatment. In this case, $\rf$, the coordinate transversal to the boundary $\partial M$, is the usual radial coordinate $r$.



We can perform a limiting procedure to find the action at null infinity, by taking a family of hypersurfaces $\Sigma_R$, indexed by a parameter $R$, such that $\Sigma_R \rightarrow \Ib^+$ as $R\rightarrow +\infty$. As an illustration, we describe two different families of hypersurfaces that can be taken: one time-like and the other space-like. For the first family, take $\Sigma_R := \{r= R ; u \geq -R \}$, where $u$ is the parameter of the null geodesics generating $\Ib^+$ and taken to be defined globally in the bulk of spacetime. This family of slices are time-like, and converge point-wise to $\Ib^+$. The other useful family of hypersurfaces, used for example in \cite{Campiglia:2016hvg, Campiglia:2016efb, Peraza:2023ivy}, is given by space-like slices $\Xi_{R} := \{t = R \}$, where $t=r+u$; these hypersurfaces correspond to the slices at constant time. In the limits $R \rightarrow +\infty$, $\Xi_{R} \rightarrow \Ib^+$. See \autoref{Slices_draw}. One can take also null hypersurfaces that approach null infinity. Recent results concerning this limit can be found in \cite{Ciambelli:2025mex}.

\begin{figure}
    \centering
    \includegraphics[width=0.4\linewidth]{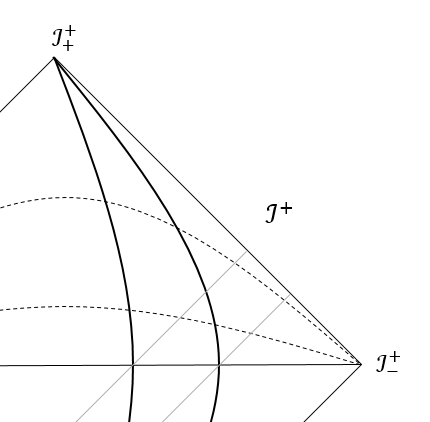}
    \caption{Representation of the two families of slices commonly used to define the symplectic structure at null infinity. Dashed line: $t=cnt.$ hypersurfaces. Solid line: $r=cnt.$ hypersurfaces. Light lines: light rays.}
    \label{Slices_draw}
\end{figure}

In what follows, we take the first family of hypersurfaces, i.e. the time-like hypersurfaces $\Sigma_R$. We have a well-defined expression for the symplectic potential at any given $R > 0$. Let $S_R = \partial \Sigma_R$. Given a pre-symplectic potential, $\theta_R$, we have that it has the following ambiguities, 

\be 
\theta_R  \mapsto \theta_R + d \pi_R + \delta \beta_R,
\ee
where the two-form $\pi_R$ modifies the canonical expression of the corner charges, i.e. it represents the addition of a corner term, and the three-form $\beta_R$ corresponds to a choice of boundary action, coming from an addition of a corner term to the Lagrangian \cite{Freidel:2019ohg}.

From the variation of the bulk and boundary Lagrangians, we can extract the associated symplectic potentials\footnote{Note that \eqref{variation_bdy} is the same as \eqref{var2} without the renormalisation.},
\beq 
\delta \mathcal{L}_{bulk} &=& - \tr \left( \delta \mathbfcal A \wedge \dd_{\mathbfcal A } * \mathbfcal F \right) + \dd \underbrace{ \tr \left(  - \delta \mathbfcal A  \wedge * \mathbfcal F \right)}_{\theta_R^{\Sigma_R}} \label{variation_bulk} \\
\delta \mathcal{L}_{bdy} &=& \tr \left( - \delta \mathbfcal A  \wedge \mathbf j -  \delta s s^{-1} ( \dd \mathbf j - [\dd ss^{-1} ,\mathbf j] ) \right) + \dd \underbrace{ \tr \left( \mathbf j \delta s s^{-1}\right)}_{\theta_R^{S_R}} \label{variation_bdy}.
\eeq

In \autoref{sec: Lagrangian}, the symplectic potential in \eqref{eq:symplectic_potential} is already a corner term, given by the renormalized version of 
\be 
\theta_R^{S_R} := \tr (\mathbf j \delta s s^{-1}),
\ee
in the limit $R \rightarrow +\infty$, and this gives the correct non-trivial charges corresponding to the sub$^n$-leading soft theorems \cite{Nagy:2024jua}. In more generality, we can extend this phase space in order to also accommodate for the left acting gauge symmetry introduced in \autoref{footnote full theta}, which has a vanishing charge. We will show the renormalization procedure for this extended \emph{pre}-symplectic potential,
\be 
\theta_R = \theta_R^{\Sigma_R} + \dd \theta_R^{S_R},
\ee
which contains both the bulk and the boundary contributions to the symplectic structure, see equations \eqref{variation_bulk} and \eqref{variation_bdy}.

By pulling back into the coordinate chart $(r,\textbf{y})$, we can write both symplectic potentials as densities. 
\beq 
\delta \mathcal{L}_{bulk} &=& - \tr \left( \delta \mathbfcal A \wedge \dd_{\mathbfcal A} * \mathbfcal F \right) + \partial_r \theta_R^{\Sigma_R \; r} + \partial_{\textbf{y}} \theta_R^{\Sigma_R \; \textbf{y}} \label{variation_bulk_1} \\
\delta \mathcal{L}_{bdy} &=& \tr \left( - \delta \mathbfcal A \wedge \mathbf j -  \delta s s^{-1} ( \dd \mathbf j - [\dd ss^{-1} , \mathbf j] ) \right) + \partial_\textbf{y} \theta_R^{S_R \; \textbf{y}} \label{variation_bdy_1}
\eeq

Observe that $\theta_R^{\Sigma_R \; r}$ is the \emph{normal} component of the symplectic potential relative to the hypersurfaces $\Sigma_R$. Then, in order to renormalize the complete symplectic potential, we proceed in two steps, first we renormalize $\theta_R^{\Sigma_R \; r}$ from \eqref{variation_bulk_1}, and then we renormalize $\theta_R^{S_R \; \textbf{y}}$ from \eqref{variation_bdy_1}.

The variation of the bulk Lagrangian can be written on-shell as follows
\be \label{eq:normal_eq}
\partial_r \theta_R^{\Sigma_R \; r} = \delta \mathcal{L}_{bulk} - \partial_{\textbf{y}} \theta_R^{\Sigma_R \; \textbf{y}}.
\ee
Then, in an expansion order by order in $r$, we can write any divergent term in $\theta_R^{\Sigma_R \; r}$ as coming from a total variation or a total divergence \cite{Freidel:2019ohg}. Indeed, if we write 
\be 
\theta_R^{\Sigma_R \; r} = \theta^{\Sigma_R \; r}_{R,(0)} (\textbf{y}) + \sum_{i=1}^{\infty} r^i Y_i (\textbf{y}) + o_R,
\ee
where $o_R$ vanishes in the limit $R \rightarrow +\infty$. Then we can take the renormalized normal component of the symplectic potential as
\be \label{eq:theta_ren}
\theta^{\Sigma_R \; r}_{ren} := \theta_R^{\Sigma_R \; r} - \sum_{i=1}^{\infty} r^i Y_i (\textbf{y}) - X(\textbf{y}),
\ee
where we can use \eqref{eq:normal_eq} to express $Y_i$ as a combination of total variations and total derivatives,
\be 
Y_i (\textbf{y}) = \text{Finite part} \left( \lim_{r \rightarrow +\infty} \frac{1}{r^{i-1}} \left( \delta \mathcal{L}_{bulk} - \partial_{\textbf{y}} \theta_R^{\Sigma_R \; \textbf{y}} \right) \right),
\ee
and $X(\textbf{y})$ is a total derivative on $\Sigma_R$, to be determined.

Next, consider equation \eqref{variation_bdy_1}. Although this equation is written in coordinates on $\Sigma_R$, it does implicitly depend on $r$ due to the location of the boundary $\partial \Sigma_R$. Then, using \eqref{jF identification} and \eqref{variation_bulk} and adding on both sides $\theta^{\Sigma_R \; r}_{R,ren} $, we obtain

\be \label{theta_bdy_div}
\theta^{\Sigma_R \; r}_{R,ren} + \partial_{\textbf{y}} \theta_R^{S_R \; \textbf{y}} = \delta \mathcal{L}_{bdy} + (\theta^{\Sigma_R \; r}_{R,ren} - \theta^{\Sigma_R \; r}) .
\ee
Observe that the first term on the left hand side is already renormalized on $r$. Inside the term $(\theta^{\Sigma_R \; r}_{R,ren} - \theta_R^{\Sigma_R \; r})$, on the right hand side, there are only higher powers in $r$, and we have that each coefficient in $r$ is either a total derivative or a total variation, cf. \eqref{eq:theta_ren}. Then we can simply subtract any higher powers of $r$ in order to renormalize $\partial_{\textbf{y}} \theta_R^{S_R \; \textbf{y}}$, 

\be 
\partial_{\textbf{y}} \theta^{S_R \; \textbf{y}}_{R,ren} := \partial_{\textbf{y}} \theta_R^{S_R \; \textbf{y}} - [\delta \mathcal{L}_{bdy} - (\theta^{\Sigma_R \; r}_{R,ren} - \theta_R^{\Sigma_R \; r})]^{O(r) \text{ and higher powers in $r$}},
\ee
as they cancel any divergent term in $r$. Finally, we arrive at 
\be 
\theta_{R,ren} := \theta^{\Sigma_R \; r}_{R,ren} + \partial_{\textbf{y}} \theta^{S_R \; \textbf{y}}_{R,ren},
\ee
which is a completely finite expression on the limit $r \rightarrow +\infty$. 

There is still a possible source of divergences: when the chart $\textbf{y}$ gets closer to spatial infinity, and therefore the coordinate generating the null geodesics goes to $\pm \infty$. For this, we denote the null coordinate $u$, and assume it is well defined in an open neighbourhood of spatial infinity (as it is the case with Bondi coordinates). The symplectic potential on the corner will be $\theta_R^{S_R \; u}$. The celestial sphere coordinates are denoted by $z,\zb$. Then, we can use the function $X(\textbf{y})$ from above to give a prescription on how to renormalize in the limit $u \rightarrow\pm \infty$. Indeed,
\be 
\theta_{R,ren} = \theta^{\Sigma_R \; r}_{R,(0)} (\textbf{y}) - X(\textbf{y}) + \partial_{\textbf{y}} \theta^{S_R \; \textbf{y}}_{R,ren},
\ee
and we can decompose
\beq 
\partial_{\textbf{y}} \theta^{S_R \; \textbf{y}}_{R,ren} &=&  \partial_{u} \theta^{S_R \; u}_{R,ren}+  \partial_{z} \theta^{S_R \; z}_{ren} +  \partial_{\zb} \theta^{S_R \; \zb}_{R,ren} \\
X(\textbf{y}) &=& \partial_u \xi^u + \partial_z \xi^z +\partial_{\zb} \xi^{\zb},
\eeq
for some vector $\xi^\textbf{y}$. Following an analogous procedure as we did for the coordinate $r$, we can construct counter-terms for $\theta_R^{S_R \; u}$ that are either total variations or total derivatives, and such that the final result is finite in the limit $u \rightarrow \pm \infty$.

\section{Fibre bundle derivation}
\label{sec: Fibre bundle}
In this section we provide a more formal derivation of the existence and transformation rules of the dressing fields introduced in \cref{sec:transformation_at_group_level}. Our approach employs the language of fibre bundles, relating the existence of the Stueckelberg field to the notion of extension/reduction of the structure group of a principal fibre bundle. To this end, we begin by briefly reviewing the basic concepts about fibre bundles, and we fix the notation that will be used throughout this section. We refer the reader to \cite{Fran_ois_2021,Nakahara:2003nw,Kobayashi:1963,Husemoller:1993,Sontz:2015,Hamilton:2017gbn,Schuller:LectureNotes2013,Attard_2018,Popov:2021gsa,Popov:2022ytx} for some references, and to \cref{App:{Lie group actions on manifolds}} for essential notions and notation about Lie group actions on manifolds.

\subsection{Review of fibre bundle formulation of gauge theories}
\subsubsection{Principal and associated bundles}\label{subsec:principal and associated bundles}
Fibre bundles provide the most comprehensive mathematical framework for describing gauge theories in physics. After recalling the definition of bundles and fixing the notation for their sections and morphisms, we will focus on two classes of bundles that play a central role in gauge theory, namely principal bundles and the associated bundles constructed from them. Precise definitions will be given below.

A \textit{bundle} is a triple $(E,\pi,M)$ where $E,M$ are topological manifolds (called the \textit{total} space and the \textit{base} space, respectively), which here are assumed to be differentiable, and $\pi\colon E\rightarrow M$ is a surjection, called the \emph{projection}. The preimage of $\pi$ at a point $m\in M$ is the \textit{fibre} of $\pi$ over $m$, and a differentiable map $\sigma\colon M\rightarrow E$ satisfying $\sigma^*\pi=\mathrm{id}_M$ is a \textit{(global)} \textit{section}\footnote{A \textit{local section} is a differentiable map $\sigma\colon U\rightarrow E$ such that $\sigma^*\pi=\mathrm{id}_U$, where $U\subset M$ is an open subset.}, where $*$ denotes the pullback. The space of all smooth sections of the bundle $(E,\pi,M)$ is denoted by $\Gamma(E)$. A \textit{fibre bundle} with (typical) fibre the manifold $F$ is a bundle in which every fibre is diffeomorphic to $F$.

Given two bundles $(E,\pi,M)$ and $(E',\pi',M')$, a \textit{bundle morphism} is a pair $(u,f)$ of maps $u\colon E\rightarrow E'$ and $f\colon M\rightarrow M'$ such that $\pi^*f=u^*\pi'$. If both $u$ and $f$ are diffeomorphisms, then $(u,f)$ is a \textit{bundle isomorphism} and we write $(E,\pi,M)\cong_\mathrm{bdl}(E',\pi',M')$. Now, consider two bundles $(E,\pi,M)$ and $(E',\pi',M)$ with the same base manifold. Then a function $u\colon E\rightarrow E'$ is a \textit{vertical bundle (iso)morphism} if $(u,\mathrm{id}_M)$ is a bundle (iso)morphism\footnote{The map $u$ is equivalently referred to as \textit{base preserving bundle (iso)morphism}, or as a \textit{bundle $M$-(iso)morphism}.}. Note that, given a bundle $\xi=(E,\pi,M)$, the set of \textit{vertical bundle automorphisms}
\begin{equation}\label{def:bdlAutvxi}
    \mathrm{Aut}^v_\text{bdl}(\xi)\coloneqq\{\Phi\colon E\rightarrow E\mid \Phi^*\pi=\pi\}
\end{equation}
forms a group under function composition.

Among fibre bundles, a particularly important class in gauge theory is that of principal bundles. These serve as the geometric framework for connections and gauge transformations, and provide the foundation on which associated bundles are constructed.
\begin{definition}[principal bundle]
    Let $G$ be a Lie group. A principal $G$-bundle is a bundle $(P,\pi,M)$ such that $P$ is a right $G$-manifold, the right action of $G$ on $P$ is free and $(P,\pi,M)\cong_{\mathrm{bdl}}(P,\rho,P/G)$, where $\rho\colon P\rightarrow P/G$ is the quotient map. The group $G$ is called the structure group of $(P,\pi,M)$.
\end{definition}
From this definition it follows that a principal $G$-bundle is a fibre bundle with typical fibre (the orbit under the right $G$-action) diffeomorphic to the structure group $G$. In order to compare different principal bundles, or to describe maps between them that respect their geometric structure, one needs the notion of a principal bundle morphism. This concept generalizes the definition of bundle morphisms to the principal setting.
\begin{definition}[principal bundle (iso)morphism]
    Let $(P,\pi,M)$ and $(P',\pi',M')$ be principal $G$- and $H$-bundles, respectively. Let $\rho\colon G\rightarrow H$ be a Lie group homomorphism and $u\colon P\rightarrow P'$, $f\colon M\rightarrow M'$ be smooth maps. The pair $(u,f)$ is called a principal bundle morphism if it is a bundle morphism and $u$ is $\rho$-equivariant. $(u,f)$ is called a principal bundle isomorphism if it is a principal bundle morphism, a bundle isomorphism and $\rho$ is a Lie group isomorphism. In this case we write $(P,\pi,M)\cong_\mathrm{p.bdl}(P',\pi',M')$.
\end{definition}
As in the case of vertical bundle morphisms, if two principal bundles have the same base manifold, one can consider \textit{vertical principal bundle (iso)morphisms}, i.e.\ principal bundle (iso)morphisms of the form $(u,\mathrm{id}_M)$. In particular, a central role in physics is played by the \textit{vertical principal bundle automorphisms} of a principal bundle $\xi=(P,\pi,M)$, that is the set
\begin{equation}\label{def:Autvxi}
    \mathrm{Aut}^v_\text{p.bdl}(\xi)\coloneqq\{\Phi\in\mathrm{Aut}^v_\mathrm{bdl}(\xi)\mid R^*_g\Phi=R_g\Phi \quad\forall g\in G\} \;,
\end{equation}
which is a group with respect to composition of functions. In the above definition, $R_g$ is the right $G$-action map associated to $g$, see \eqref{def:action maps}\footnote{Using a more explicit notation, \eqref{def:Autvxi} can be written as $\mathrm{Aut}^v_\mathrm{p.bdl}(\xi)\coloneqq\{\Phi\colon P\rightarrow P\mid \pi\circ\Phi=\pi\,,\ \ \Phi(p\ra g)=\Phi(p)\ra g \ \ \forall p\in P,g\in G\}$, where $\ra$ is the right $G$-action on $P$. See \cref{App:{Lie group actions on manifolds}} for more details.}. We emphasize that the groups \eqref{def:bdlAutvxi} and \eqref{def:Autvxi} do not generally coincide, because in the latter case one requires, in addition, compatibility with the right $G$-action.

Given a principal bundle, one can construct new fibre bundles whose fibres carry a representation of the structure group. These are called associated bundles, and they allow one to translate the geometric data of the principal bundle into other contexts, such as vector bundles or tensor bundles. 
\begin{definition}[associated bundle]\label{def:associated bundle}
    Let $\xi=(P,\pi,M)$ be a principal $G$-bundle and $F$ a smooth left $G$-manifold. The associated bundle to $\xi$ with fibre $F$ is a bundle $\xi[F]\coloneqq(P_F,\pi_F,M)$ with total space $P_F\coloneqq(P\times F)/\sim_G$ (sometimes written as $P\times_G F$), where
    \begin{equation}
        (p,f)\sim_G(p',f') \quad \colon\Leftrightarrow \quad \exists g\in G \ \text{such that} \ \begin{cases}
            p'=R_gp \\
            f'=L_{g^{-1}} f
        \end{cases}\;,
    \end{equation}
    and projection $\pi_F\colon P_F\rightarrow M$ defined\footnote{We denote the elements of $P_F$ using square brackets, since they are equivalence classes.} as $\pi_F([p,f])\coloneqq\pi(p)$.
\end{definition}
In the above definition, $R_g$ is the right $G$-action map on $P$, while $L_g$ is the left $G$-action map on $F$. It is easy to see that an associated bundle $\xi[F]$ is a fibre bundle with typical fibre $F$. We now recall a standard result (see, e.g., \cite{Husemoller:1993}) that establishes a correspondence between sections of an associated bundle and maps from the total space of the underlying principal bundle to the fibre, satisfying a suitable compatibility condition.
\begin{proposition}\label{th:function associated to sections}
    Let $\xi=(P,\pi,M)$ be a principal $G$-bundle, $F$ be a left $G$-manifold and
    \begin{equation}
        \label{def:calU}
        \mathcal{U}\coloneqq\{u\colon P\rightarrow F\mid R_g^*u=L_{g^{-1}}u \quad \forall g\in G\}\;,
    \end{equation}
    where $L$ is the left $G$-action on $F$. Then $\mathcal{U}\xleftrightarrow{1:1}\Gamma(\xi[F])$ via
    \begin{equation}
        \sigma_{\bullet}\colon\mathcal{U}\rightarrow\Gamma(\xi[F])\,, \quad u\mapsto\sigma_u \;,
    \end{equation}
    where
    \begin{equation}
        \sigma_u\colon M\rightarrow P\times_G F\,, \quad [p]_G\mapsto [p,u(p)]_G\;.
    \end{equation}
\end{proposition}
We conclude this section by highlighting a particular type of associated bundle that will play a role later on. For a given principal bundle, one can construct its adjoint bundle by taking the structure group as the fibre and letting the group act on itself via the adjoint action. This construction yields a canonical associated bundle that will provide, in the following subsection, an alternative viewpoint on gauge transformations.
\begin{definition}[adjoint bundle]\label{def:adjoint bundle}
    The adjoint bundle $\mathrm{Ad}(\xi)$ of a principal $G$-bundle $\xi$ is the associated bundle $\xi[G]$ with left $G$-action $L_g\coloneqq \mathrm{Ad}_g$\footnote{In this paper, we use the convention $\mathrm{Ad}_gg'\coloneqq gg'g^{-1}$, where $g,g'\in G$.}.
\end{definition}

\subsubsection{Gauge transformations as vertical automorphisms}\label{sec:gauge transf}
Having provided the necessary background to recall the basic concepts about group actions and fibre bundles, we now turn to the discussion of gauge transformations within this framework. For a detailed treatment and further references on gauge transformations in the setting of fibre bundles, see for instance \cite{Husemoller:1993,Attard_2018}.
\begin{definition}[gauge transformation]
    A gauge transformation of a principal $G$-bundle $\xi$ is an element of $\mathrm{Aut}^v_\mathrm{p.bdl}(\xi)$.
\end{definition}
In this perspective, a gauge transformation of a principal bundle is simply a vertical principal bundle automorphism. It reshuffles the points in each fibre according to the structure group, leaving the underlying geometric and topological structure unchanged.
An alternative yet equivalent description of gauge transformations is given by the following standard result.
\begin{proposition}\label{gauge transf interpretations}
    Let $\xi=(P,\pi,M)$ be a principal $G$-bundle, and
    \begin{equation}
        \label{def:gauge group G}
        \mathcal{G}\coloneqq\{\gamma\colon P\rightarrow G\mid R^*_g\gamma=\mathrm{Ad}_{g^{-1}}\gamma \quad \forall g\in G\}\;,
    \end{equation}
    which is a group with respect to pointwise multiplication. Then $\mathrm{Aut}^v_\mathrm{p.bdl}(\xi)\cong\mathcal{G}$ via
    \begin{equation}\label{def:isomorph Autv G}
        \Phi_{\bullet}\colon\mathcal{G}\rightarrow\mathrm{Aut}^v_\mathrm{p.bdl}(\xi) \,, \quad \gamma\mapsto\Phi_\gamma\;,
    \end{equation}
    where $\Phi_\gamma(p)\coloneqq R_{\gamma(p)}p$ for all $p\in P$.
\end{proposition}
By virtue of this isomorphism, the group $\mathcal{G}$ is commonly referred to as the \textit{gauge group} of the principal bundle $\xi$. At this point, it is convenient to pause for a moment and emphasize the difference between the structure group $G$ and the gauge group $\mathcal{G}$ of $\xi$.

On the one hand, $G$ is part of the defining data of the principal bundle and it specifies the intrinsic symmetry of the fibres of $\xi$. In physical terms, $G$ determines the kind of gauge field; for instance, one can choose the structure group to be $U(1)$ (for electromagnetism) or $SU(2)$ (for weak interactions). In the case relevant to our discussion, the structure group will be taken to be a particular loop group $LG$ that will be defined below. This choice allows us to construct a suitable principal bundle whose base space is the boundary of spacetime, thus providing the natural geometric setting for studying gauge fields on the boundary. This concept will be developed in detail in \cref{sec:Extensions and reductions in presence of Boundaries}.

On the other hand, the gauge group corresponds to fibre-preserving automorphisms that encode different (yet physically equivalent) ways of representing the same field configuration. In other words, $\mathcal{G}$ encodes the local redundancies, since gauge transformations change the local representation of the fields along the fibres without affecting observables. More precisely, a gauge transformation $\gamma\in \mathcal{G}$ acts both on a connection one-form $\omega$ and on a curvature two-form $\Omega$ via pullback along the corresponding vertical automorphism $\Phi_\gamma\in\mathrm{Aut}^v_\mathrm{p.bdl}(\xi)$,
\begin{equation}\label{def:gauge transf action}
    \omega^\gamma\coloneqq\Phi^*_\gamma\omega=\mathrm{Ad}_{\gamma^{-1}}\circ\omega+\gamma^{-1}\dd \gamma\,, \quad \Omega^\gamma\coloneqq\Phi_\gamma^*\Omega=\mathrm{Ad}_{\gamma^{-1}}\circ \Omega\;.
\end{equation}
This is a classical result, which can be found in any standard reference on the subject (see, for instance, \cite{Nakahara:2003nw}). In particular, note that $\omega$ acquires an inhomogeneous term under the transformation, whereas $\Omega$ transforms covariantly, so that physically relevant quantities built from it (such as traces and characteristic classes) remain unaffected. Moreover, as expected, the application of two successive gauge transformations $\Phi_\gamma,\Phi_\eta$ can be expressed as the action of a single gauge transformation $\Phi_{\gamma\eta}$. This comes directly from the fact that the composition law of $\mathrm{Aut}^v_\mathrm{p.bdl}$ corresponds to the product in $\mathcal{G}$. Since the gauge group acts via pullback we have, for instance for a connection, that
\begin{equation}\label{eq:composition law gauge transf}
    (\omega^\gamma)^\eta=\Phi_\eta^*(\Phi_\gamma^*\omega)=(\Phi_\gamma\circ\Phi_\eta)^*\omega=\Phi_{\gamma\eta}^*\omega=\omega^{\gamma\eta}\;.
\end{equation}

Let us pause for a moment to make the following remark. In this framework, gauge transformations have a non trivial action on themselves, once again via pullback. This means that the gauge group admits a natural action on its own elements: given $\gamma,\eta\in\mathcal{G}$, the transformation of $\gamma$ by $\eta$ is defined as\footnote{The last equality in \eqref{def:left calG on itself} comes from $\Phi^*_\eta\gamma(p)=\gamma(\Phi_\eta(p))=\gamma(p\ra\eta(p))=\mathrm{Ad}_{\eta^{-1}(p)}\gamma(p)$, for all $p\in P$.}
\begin{equation}\label{def:left calG on itself}
    \gamma^\eta\coloneqq\Phi^*_\eta\gamma=\mathrm{Ad}_{\eta^{-1}}\gamma\;.
\end{equation}
This action is compatible with the composition law \eqref{eq:composition law gauge transf} and the explicit form of the transformed connection \eqref{def:gauge transf action}. Indeed:
\begin{equation}\label{eq:explicit computation composition gauge transf}
    \begin{aligned}
        (\omega^\gamma)^\eta &=(\omega^\eta)^{\gamma^\eta}=\mathrm{Ad}_{(\gamma^\eta)^{-1}}\omega^\eta+(\gamma^\eta)^{-1}\dd\gamma^\eta \\
        &=\mathrm{Ad}_{(\gamma^\eta)^{-1}}(\mathrm{Ad}_{\eta^{-1}}\omega+\eta^{-1}\dd\eta)+(\gamma^\eta)^{-1}\dd\gamma^\eta \\
            &=\mathrm{Ad}_{(\eta\gamma^\eta)^{-1}}\omega+(\eta\gamma^\eta)^{-1}\dd(\eta\gamma^\eta)=\omega^{\eta\gamma^\eta}=\omega^{\eta\mathrm{Ad}_{\eta^{-1}}\gamma}=\omega^{\gamma\eta}\;.
    \end{aligned}
\end{equation}
At first sight, the fact that gauge transformations act non-trivially on themselves may seem wrong, since in the standard formulation of gauge theories the gauge parameter is usually regarded as non-transforming. This discrepancy, however, does not imply an incompatibility between the two descriptions. Rather, it simply reflects a different convention for the composition rule of gauge transformations. Indeed, the usual computation for two successive gauge transformations with non-transforming parameter gives
\begin{equation}
        (\omega^\gamma)^\eta=(\mathrm{Ad}_{\gamma^{-1}}\circ\omega+\gamma^{-1}\dd\gamma)^\eta=\mathrm{Ad}_{\gamma^{-1}}\circ\omega^\eta+\gamma^{-1}\dd\gamma=\omega^{\eta\gamma}\;,
\end{equation}
which, compared with \eqref{eq:explicit computation composition gauge transf}, amounts to a different sign convention at the level of commutators of infinitesimal gauge transformations.

Finally, let us notice that the gauge group $\mathcal{G}$ of a principal bundle $\xi$ is in one to one correspondence with the space of sections of the adjoint bundle of $\xi$, that is, $\mathcal{G}\xleftrightarrow{1:1}\Gamma(\mathrm{Ad}(\xi))$. This correspondence is an immediate consequence of \cref{th:function associated to sections} and \cref{def:adjoint bundle}, and provides yet another way to interpret gauge transformations.
    



\subsubsection{Reduction and extension of structure group}
In the study of gauge theories from the perspective of fibre bundles, it can be useful to compare connections defined on principal bundles with different structure groups. To be more precise, starting from a principal bundle with structure group $G$, one may be interested in building a principal bundle whose structure group arises as either a reduction or an extension of $G$. As we shall see below, this situation provides the natural setting in which dressing fields, together with their transformation laws, appear.

To start with, let us recall the precise notion of principal bundles related by reduction or extension of their structure groups.
\begin{definition}[reduction/extension]\label{def:reduction-extension}
    Let $H\subseteq G$ be Lie groups. Let $\xi=(P,\pi,M)$ and $\xi'=(P',\pi',M)$ be principal $G$- and $H$-bundles, respectively. Then $\xi'$ is called an $H$-reduction of $\xi$ (and $\xi$ is a $G$-extension of $\xi'$) if there exists a bundle morphism $(u,f)$ from $\xi'$ to $\xi$ such that
    \begin{equation}
        R_h^*u=R_hu \quad \forall h\in H \;.
    \end{equation}
\end{definition}

A classical result \cite{Kobayashi:1963,Husemoller:1993} then establishes precise conditions for the existence of such reductions and extensions, ensuring that these constructions are not only formal but can be realized in concrete situations. Before stating it, it is convenient to recall a preliminary fact that will also be useful in the next section. Starting from a principal bundle, one can construct a new bundle by quotienting the total space by a chosen subgroup of the structure group. The resulting bundle can be shown to be a fibre bundle.
\begin{lemma}\label{lemma:isomorphism xiH and xiGH}
    Let $H\subseteq G$ be Lie groups and $\xi=(P,\pi,M)$ be a principal $G$-bundle. Then the bundle $\xi/H\coloneqq(P/H,\pi',M)$ with projection $\pi'$ such that $\pi=\pi'\circ\rho_H$, where $\rho_H\colon P\rightarrow P/H$ is the quotient map, is a fibre bundle with typical fibre $G/H$. Moreover, $\xi/H\cong_{\mathrm{bdl}}\xi[G/H]$ via vertical isomorphism.
\end{lemma}
This lemma provides the basic construction needed to study the problem of reducing or extending the structure group. In particular, the existence of such reductions can be formulated in terms of sections of the quotient bundle.

\begin{proposition}[existence of reduction/extension]\label{th:existence of reduction/extension}
    Let $H\subseteq G$ be Lie groups. Then:
    \begin{enumerate}[label=\roman*)]
        \item any principal $H$-bundle admits a $G$-extension;
        \item a principal $G$-bundle $\xi$ admits an $H$-reduction if and only if the bundle $\xi/H$ has a global well-defined section.
    \end{enumerate}
\end{proposition}

\subsection{The dressing field method}
We now introduce the key tool that encodes the Stueckelberg procedure in the framework of fibre bundles, namely the dressing field method, which has been extensively discussed in the literature \cite{Attard_2018,Fran_ois_2021,Popov:2021gsa,Popov:2022ytx,Fournel:2012cr,Francois:2015oca,Attard:2016stx,Francois:2017akk,Mathieu:2019lgi,Berghofer:2021ufy,FrancoisAndre:2023jmj,Berghofer:2024plf,Francois:2025sic}. This procedure provides a systematic way to build composite gauge fields with modified transformation properties. A dressing field is a map on the total space of a principal bundle $\xi$ with values in a subgroup of the structure group, with a prescribed equivariance. Such a field can be used to dress the connection on $\xi$, so that the resulting dressed field is strictly invariant under the given subgroup, or transforms covariantly under a residual or even extended symmetry group. From the bundle-theoretic perspective, this method can be interpreted as producing a connection on either a reduced bundle or an extended one.

\begin{definition}[dressing/dressed fields]
    Let $\xi=(P,\pi,M)$ be a principal $G$-bundle, $\omega$ a connection one-form on $\xi$ and $K\subseteq J\subset G$ Lie subgroups. We define the set of $(J,K)$-dressing fields as
    \begin{equation}
        \label{def:dressing fields}
        \mathrm{Dr}(J,K)\coloneqq\{u\colon P\rightarrow J\mid R^*_ku=L_{k^{-1}}u \quad\forall k\in K\}\;.
    \end{equation}
    For a given dressing field $u\in\mathrm{Dr}(J,K)$, let $\phi_u\in\mathrm{Aut}^v_\mathrm{bdl}(\xi)$ be the map\footnote{Note that $\phi_u$ is well defined, since for all $u\in\mathrm{Dr}(J,K)$ and $p\in P$ we have $\phi^*_u\pi(p)=\pi(p\ra u(p))=\pi(p)$ (where $\ra$ is the right $G$-action of the principal bundle), ensuring that $\phi_u\in\mathrm{Aut}^v_\mathrm{bdl}(\xi)$.} $\phi_u(p)\coloneqq R_{u(p)}p$. Then, the $u$-dressed field is defined to be the pullback of $\omega$ by $\phi_u$,
    \begin{equation}
        \omega^u\coloneqq \phi^*_u\omega=\mathrm{Ad}_{u^{-1}}\circ\omega+u^{-1}\dd u \;.
    \end{equation}
\end{definition}
We observe\footnote{Cf. the defining transformation properties of \eqref{def:gauge group G} and \eqref{def:dressing fields}.} that $\mathrm{Dr}(J,K)\not\subset\mathcal{G}$, where $\mathcal{G}$ is the gauge group of $\xi$. In particular, this implies that, in general, $\phi_u$ is not a vertical principal bundle automorphism, but only a vertical bundle automorphism, cf. equations \eqref{def:bdlAutvxi} and \eqref{def:Autvxi}. Moreover, note that $u$ can be equivalently defined by its gauge transformation rule,
\begin{equation}
    \label{def:dressing fields def by transf}
    u^\kappa=\Phi^*_\kappa u\coloneqq L_{\kappa^{-1}}u \quad \forall\kappa\in\mathcal{K}\;,
\end{equation}
where
\begin{equation}
    \label{def:gauge subgroup K}
    \mathcal{K}\coloneqq\{\kappa\colon P\rightarrow K\mid R^*_k\kappa=\mathrm{Ad}_{k^{-1}}\kappa \quad \forall k\in K\}
\end{equation}       
is the \textit{gauge subgroup}\footnote{Note that $\mathcal{K}\not\subset\mathcal{G}$, since the maps in $\mathcal{K}$ are not subject to any compatibility condition with the action of elements in $G\setminus K$. However, the definitions of $\mathcal{K}$ and $\mathcal{G}$ differ only by the image and transformation groups $K\subset G$, cf. \eqref{def:gauge group G} and \eqref{def:gauge subgroup K}. For this reason, we call $\mathcal{K}$ the gauge subgroup of $\mathcal{G}$, via an abuse of terminology.} associated with $K$ and $\Phi_\kappa(p)\coloneqq R_{\kappa(p)}p$ for $p\in P$. Indeed, the fact that $\mathrm{im}(\kappa)\subseteq K$ imply that \eqref{def:dressing fields def by transf} and \eqref{def:dressing fields} are equivalent definitions for the dressing field $u$. At a point $p\in P$, the transformation reads $u^\kappa(p)=L_{\kappa(p)^{-1}}u(p)$. Notice that, since $\kappa\in \mathcal{K}\neq\mathcal{G}$, the map $\Phi_\kappa\colon P\rightarrow P$ defined above is not a gauge transformation (i.e. a vertical principal bundle automorphism) of $\xi$, it is simply a vertical bundle automorphism of $\xi$.

Due to the form of \eqref{def:dressing fields def by transf}, observe that we have a generalization of \eqref{eq:composition law gauge transf} with $\eta$ (but not $\gamma$) now changed to a dressing $u$,
\begin{equation}\label{eq:composition law dressing transf}
    (\omega^\gamma)^u= \omega^{\gamma u}\;.
\end{equation}

The following well-known result \cite{Attard_2018,Popov:2022ytx} states that a connection on a principal $G$-bundle dressed with a $(J,K)$-dressing field is invariant under the gauge transformations given by $K$. 

\begin{proposition}
    Let $\xi=(P,\pi,M)$ be a principal $G$-bundle with gauge group $\mathcal{G}$, $\omega$ a connection one-form on $\xi$, and $K\subseteq J\subset G$ Lie subgroups, with $\mathcal{K}$ the gauge subgroup associated with $K$. Then for all $u\in\mathrm{Dr}(J,K)$ the dressed field $\omega^u$ is $\mathcal{K}$-invariant.
\end{proposition}
\begin{proof}
    Note that $\forall p\in P,k\in K$
    \begin{equation}
        R^*_k\phi_u(p)=\phi_u(p\ra k)=p\ra k\ra u(p\ra k)=p\ra k\ra k^{-1}u(p)=\phi_u(p)\;,
    \end{equation}
    where $\ra$ is the right $G$-action on $\xi$. Equivalently, $\Phi^*_\kappa\phi_u=\phi_u$ for all $\kappa\in \mathcal{K}$. Therefore, $\forall\kappa\in\mathcal{K}$,
    \begin{equation}
        \Phi^*_\kappa\phi^*_u=(\phi_u\circ\Phi_\kappa)^*=(\Phi^*_\kappa\phi_u)^*=\phi^*_u\;,
    \end{equation}
    where the first equality follows from the composition property of pullbacks. This implies
    \begin{equation}
        (\omega^u)^\kappa=\Phi^*_\kappa\omega^u=\Phi^*_\kappa\phi^*_u\omega=\phi^*_\kappa\omega=\omega^u \;,
    \end{equation}
    i.e. $\omega^u$ is $\mathcal{K}$-invariant. An equivalent way to prove this result is the following. The transformation of the dressed field under $\kappa\in\mathcal{K}$ can be written as
    \begin{equation}
        (\omega^u)^\kappa=(\omega^\kappa)^{u^\kappa}\overset{\eqref{def:dressing fields def by transf}}{=}(\omega^\kappa)^{\kappa^{-1}u}\overset{\eqref{eq:composition law dressing transf}}{=} \omega^{\kappa (\kappa^{-1} u)} = \omega^u .
    \end{equation}
    
\end{proof}

\subsubsection{Connection on reduced bundles}

At this point, we turn to the relation between reductions of principal bundles and the emergence of dressing fields. Within this framework, we show that dressing fields appear naturally as the objects that allow us to dress the original connection in order to obtain a connection on the reduced bundle. In addition, as such dressing fields are designed to produce the new connections, they come equipped with precise transformation laws, which we derive explicitly\footnote{See \cite{Attard_2018,Fran_ois_2021} for a similar treatment.}.

As one may infer from \cref{sec:transformation_at_group_level}, and as will become clear in \cref{sec:Extensions and reductions in presence of Boundaries} using loop groups, we are particularly interested in the case where the structure group is reduced to a splitting subgroup. Concretely, if $G$ is the structure group and $N\triangleleft G$ is a normal subgroup, we focus on the situation in which there exists a subgroup $H\subseteq G$ that projects isomorphically onto $G/N$ and satisfies $G=H\ltimes N$. In order to do that, we will need the following maps.

\begin{definition}[projections on semidirect product]\label{def:projections on semidirect product}
    Let $G=H\ltimes N$ be a Lie group. We define the maps $\mathrm{pr}_N\colon G\rightarrow N$ and $\mathrm{pr}_H\colon G\rightarrow H$ via the relation
    \begin{equation}
        g=\mathrm{pr}_N(g)\mathrm{pr}_H(g) \quad \forall g\in G\;.
    \end{equation}
\end{definition}
This definition is well posed, since for all $g\in G=H\ltimes N$ there exists a unique pair $(n,h)\in N\times H$ such that $g=nh$, so $G=NH$ and $H\cong G/N$. One can easily see that the maps above satisfy, for all $n\in N$ and $h\in H$,
\begin{equation}
    \mathrm{pr}_N(n)=n\,,\quad \mathrm{pr}_H(h)=h\,, \quad\mathrm{pr}_N(h)=\mathrm{pr}_H(n)=\mathbf{1}_G\;.
\end{equation}
This follows from \cref{def:projections on semidirect product} and $H\cap N=\{\textbf{1}_G\}$. One can prove the following result.

\begin{theorem}[connection on reduction to quotient group]\label{th:connection on reduction to quotient group}
    Let $G=H\ltimes N$ be a Lie group, $\xi=(P,\pi,M)$ a principal $G$-bundle and $\omega$ a connection one-form on $\xi$. Let $\tilde\xi$ be an $H$-reduction of $\xi$. Then there exists a map $u\colon P\rightarrow N$ satisfying
    \begin{equation}
        \label{transformation of uvarphi}
        R^*_gu=L_{g^{-1}}R_{\mathrm{pr}_H(g)}u \quad \forall g\in G\;.
    \end{equation}
    The dressed connection $\omega^u$ transforms as a connection on $\tilde\xi$ under a gauge transformation $\Phi_\gamma\in\mathrm{Aut}^v_\mathrm{p.bdl}(\xi)$,
    \begin{equation}
        (\omega^u)^\gamma=(\omega^u)^\eta=\mathrm{Ad}_{\eta^{-1}}\circ\omega^u+\eta^{-1}\dd \eta\;,
    \end{equation}
    where $\eta\coloneqq\mathrm{pr}_H\circ\gamma$.
\end{theorem} 

\begin{proof}
    Since $\tilde\xi$ is an $H$-reduction of $\xi$, \cref{th:existence of reduction/extension} implies that the bundle $\xi/H$ admits a section, which induces a section on $\xi[G/H]$ via the base preserving isomorphism of \cref{lemma:isomorphism xiH and xiGH}. Thus, \cref{th:function associated to sections} ensures that there exists a function $\tilde u\colon P\rightarrow G/H$ such that $\forall g\in G$
    \begin{equation}
        \label{transf rule u tilde}
        R^*_g\tilde u=L_{g^{-1}}\tilde u\;.
    \end{equation}
    Let
    \begin{equation}
        \label{def:isomorphism as sets}
        \varphi\colon N\rightarrow G/H\,, \quad n\mapsto [n]_H
    \end{equation}
    be the bijection identifying the quotient space $G/H$ with the normal subgroup $N$ as sets. Then, there exists a unique lift of $\tilde{u}$ via the map $\varphi$, namely there exists a unique map $u\colon P\rightarrow N$ such that the following diagram commutes:
    \begin{equation}
        \begin{tikzcd}
        & N \arrow[dr, "\varphi"] & \\
        P \arrow[ur, "u"] \arrow[rr, "\tilde u"] & & G/H
        \end{tikzcd}
    \end{equation}
    Then the transformation rule \eqref{transformation of uvarphi} arises as the compatibility condition with \eqref{transf rule u tilde}: for all $g\in G$ we have
    \begin{equation}
        R^*_g\tilde u=R^*_g[u]_H=[R^*_gu]_H\overset{!}{=}L_{g^{-1}}\tilde{u}=L_{g^{-1}}[u]_H=[L_{g^{-1}}R_{\mathrm{pr}_H(g)}u]_H
    \end{equation}
    
    In particular,
    \begin{equation}
        \begin{aligned}
            R^*_nu&=L_{n^{-1}}R_{\mathrm{pr}_H(n)}u=L_{n^{-1}}u \quad \forall n\in N \\
            R^*_hu&=L_{h^{-1}}R_{\mathrm{pr}_H(h)}u=\mathrm{Ad}_{h^{-1}}u \quad \forall h\in H\;,
        \end{aligned}
    \end{equation}
    so $u$ is a $(N,N)$-dressing field. Equivalently,
    \begin{equation}
        \begin{aligned}
            u^\nu&\coloneqq R^*_\nu u=L_{\nu^{-1}}u \quad \forall \nu\in \mathcal{N} \\
            u^\eta&\coloneqq R^*_\eta u=\mathrm{Ad}_{\eta^{-1}}u \quad \forall \eta\in \mathcal{H}\;,
        \end{aligned}
    \end{equation}
    where $\mathcal{N},\mathcal{H}$ are the gauge subgroups\footnote{See \eqref{def:gauge subgroup K}.} of $\mathcal{G}$ (which is the gauge group of $\mathcal{\xi}$) associated with $N,H$ respectively. Then the dressed field $\omega^u$ transforms under a gauge transformation $\Phi_\gamma\in\mathrm{Aut}^v_\mathrm{p.bdl}(\xi)$ as
    \begin{equation}
        (\omega^u)^\gamma=(\omega^\gamma)^{u^\gamma}=(\omega^\gamma)^{\gamma^{-1}u\eta}=(\omega^u)^\eta \;.
    \end{equation}
\end{proof}

In words, this means that whenever a principal $G$-bundle can be reduced to a $(G/N)$-bundle for some normal subgroup $N$, one can use a suitable dressing field to turn any $G$-connection into a $G/N$-connection. This dressing field is a $N$-valued map on the total space of the $G$-bundle, transforming as in \eqref{transformation of uvarphi}.

However, the field $u$ is not the most general object that one may introduce for this purpose. Indeed, while in the proof above $u$ was defined as the lift of $\tilde u$ via the bijection $\varphi$ (so that $u$ naturally takes value in the subgroup $N$), one could instead lift $\tilde{u}$ via the quotient map $\rho_H\colon G\rightarrow G/H$. Since such a lift is not unique, this procedure would produce a whole family of dressing fields, now taking values in the full group $G$ and whose transformation rules differ from \eqref{transformation of uvarphi}, yet still defines a connection transforming under the quotient group $G/N$. This construction is mathematically well-defined and entirely admissible for the purpose of defining connections on the reduced bundle. Nevertheless, it turns out to be not suitable from a physical standpoint. In the setting considered in this paper, we require the dressing fields associated with a bundle reduction encode \textit{all and only} the degrees of freedom corresponding to $N$, the normal subgroup by which $G$ is reduced. A dressing field taking values in the full group $G$ rather than in $N$ would contain redundant information and lack a clear physical interpretation.

\subsection{Stueckelberg fields}

Here, we focus on a different scenario: given a connection on a principal bundle, we aim to construct a connection on an extended, rather than reduced, bundle. Among the possible ways a structure group $H$ can be extended to a group $G$, we are interested in the split case, i.e. when $G= H\ltimes N$ for some normal Lie group $N$. In this case, we say that $G$ is a split extension of $H$ by $N$. Conceptually, this situation can be seen as dual or ``inverse'' to the reduction considered in the previous section. Therefore, the fields we will use here to construct extended connections naturally arise as the ``inverses'' of those introduced in \cref{th:connection on reduction to quotient group}. We refer to these fields as \textit{Stueckelberg fields} and show that they can be equivalently described by a section of a suitably defined associated bundle, which we call the \textit{Stueckelberg bundle}. To this end, we begin with the following definition.

\begin{definition}[Stueckelberg action] 
    Let $G$ be a split extension of $H$ by $N$. We define the Stueckelberg action as the map
    \begin{equation}
        \label{def:stueckelberg action}
        \mathrm{St}\colon G\times N\rightarrow N\,, \quad (g,n)\mapsto \mathrm{pr}_H(g)ng^{-1}\;,
    \end{equation}
    where $\mathrm{pr}_H$ is the canonical projection on $H$, see \cref{def:projections on semidirect product}.
\end{definition}

\begin{remark}
    The above definition is well-posed.
\end{remark}
\begin{proof}
    For all $g\in G$ and $n\in N$,
    \vspace{-3mm}
    \begin{equation}
        \mathrm{St}(g,n)=\mathrm{pr}_H(g)\,n\,\mathrm{pr}_N(g^{-1})\mathrm{pr}_H(g^{-1})=\mathrm{pr}_H(g)\,\overbrace{n\,\mathrm{pr}_N(g^{-1})}^{\in N}\mathrm{pr}_H(g)^{-1}\in N\;,
    \end{equation}
    since $N\triangleleft G$. In the last equality we used
    \begin{equation}
        \mathrm{pr}_H(g^{-1})=\mathrm{pr}_H(g)^{-1} \quad \forall g\in G \;,
    \end{equation}
    which is true since $\mathrm{pr}_H$ is a homomorphism and, as such, maps inverses to inverses.
\end{proof}

In what follows, we will often use the notation $\mathrm{St}_g(n)\coloneqq\mathrm{St}(g,n)$, so that $\mathrm{St}_g\in\mathrm{Aut}(N)$ for all $g\in G$. We are now in a position to introduce the definition of Stueckelberg fields as $N$-valued maps on a principal $G$-bundle that transform under the Stueckelberg action. As we shall see, this transformation rule is precisely the one needed to construct connections on split extended principal bundles.

\begin{definition}[Stueckelberg fields]
    Let $\xi=(P,\pi,M)$ be a principal $G$-bundle, with $G$ a split extension of $H$ by $N$. We define the set of $(N,G)$-Stueckelberg fields as
    \begin{equation}
        \label{def:stueckelberg fields}
        \mathrm{St}(N,G)\coloneqq\{s\colon P\rightarrow N\mid R^*_gs=\mathrm{St}_{g^{-1}}s \quad \forall g\in G\}\;.
    \end{equation}
\end{definition}
By following the same logic of \cref{subsec:principal and associated bundles}, a geometric characterization of the above fields can be given in terms of a suitable associated bundle.
\begin{definition}[Stueckelberg bundle]
    Let $\xi$ be a principal $G$-bundle, with $G$ a split extension of $H$ by $N$. The Stueckelberg bundle $\mathrm{St}(\xi)$ is the associated bundle $\xi[N]$ with left $G$-action given by Stueckelberg action $\mathrm{St}$.
\end{definition}
\begin{proposition}
    Let $\xi$ be a principal $G$-bundle, with $G$ a split extension of $H$ by $N$. Then
    \begin{equation}
        \mathrm{St}(N,G)\xleftrightarrow{1:1}\Gamma(\mathrm{St}(\xi))\;.
    \end{equation}
\end{proposition}
\begin{proof}
    It follows immediately from \cref{th:function associated to sections}.
\end{proof}
In other words, the Stueckelberg fields can be equivalently regarded as sections of the Stueckelberg bundle, in much the same way that gauge transformations can be interpreted as sections of the adjoint bundle (see comment at the and of \cref{sec:gauge transf}). This identification provides a clear geometric interpretation of Stueckelberg fields, highlighting their role as natural objects associated with the underlying principal bundle structure.

\subsubsection{Connection on extended bundles}
The transformation law \eqref{def:stueckelberg fields} in our definition of a Stueckelberg field is precisely the one required to construct connections on a bundle whose structure group arises as a split extension, as we prove in the following theorem.

\begin{theorem}[connection on split extension]\label{th:connection on split extension}
    Let $\tilde\xi$ be a principal $H$-bundle with connection one-form $\omega$, and let $G$ be a split extension of $H$ by a Lie group $N$. Then there exists a Stueckelberg field $s\in\mathrm{St(N,G)}$ on the extended $G$-bundle $\xi$, such that the dressed field
    \begin{equation}
        \omega^s\coloneqq R^*_s\omega=\mathrm{Ad}_{s^{-1}}\circ\omega+s^{-1}\dd s
    \end{equation}
    transforms as a connection one-form on $\xi$ under a gauge transformation $\Phi_\gamma\in\mathrm{Aut}^v_\mathrm{p.bdl}(\xi)$,
    \begin{equation}
        (\omega^s)^\gamma=\mathrm{Ad}_{\gamma^{-1}}\circ\omega^s+\gamma^{-1}\dd \gamma\;.
    \end{equation}
\end{theorem}
\begin{proof}
    Proposition \ref{th:existence of reduction/extension} guarantees that a $G$-extension $\xi$ of the $H$-bundle $\tilde\xi$ always exists. In turn, $\tilde\xi$ is an $H$-reduction of such extension $\xi$, so \cref{th:connection on reduction to quotient group} ensures the existence of a map $u\colon P\rightarrow N$ such that
    \begin{equation}
        R^*_gu=L_{g^{-1}}R_{\mathrm{pr}_H(g)}u \quad \forall g\in G\;.
    \end{equation}
    We define the map $s\colon P\rightarrow N$ to be the pointwise inverse of $u$,
    \begin{equation}
        \label{pf:stueckelberg field}
        s(p)\coloneqq u(p)^{-1} \quad \forall p\in P \;.
    \end{equation}
    This map transforms under an element $g\in G$ as
    \begin{equation}
        \begin{aligned}
            R^*_gs&=R^*_g(u)^{-1}=(R^*_gu)^{-1}=L_{\mathrm{pr}_H(g)^{-1}}R_{g}(u)^{-1}=L_{\mathrm{pr}_H(g^{-1})}R_{g}s=\mathrm{St}_{g^{-1}}s\;,
        \end{aligned}
    \end{equation}
    so it is a $(N,G)$-Stueckelberg field. Under a gauge transformation $\Phi_\gamma\in\mathrm{Aut}^v_\mathrm{p.bdl}(\xi)$, it transforms as
    \begin{equation}
        s^\gamma\coloneqq R^*_\gamma s=L_{\eta^{-1}}R_\gamma s \;,
    \end{equation}
    where $\eta=\mathrm{pr}_H\circ\gamma$. Thus, the composite field $\omega^s$ transforms as
    \begin{equation}
        (\omega^s)^\gamma=(\omega^\gamma)^{s^\gamma}=(\omega^\eta)^{\eta^{-1}s\gamma}=\omega^{s\gamma}\;,
    \end{equation}
    for an arbitrary gauge transformation $\gamma\in\mathcal{G}$ in the gauge group of $G$.
\end{proof}
The content of this theorem is conceptually analogous to that of \cref{th:connection on reduction to quotient group}: we begin with a connection on a principal bundle, modify the structure group, and then identify new fields with suitable transformation properties, which are subsequently used to define connections transforming appropriately under the gauge group associated to the new structure group. There is, however, a significant difference between the two cases. In \cref{th:connection on reduction to quotient group}, the existence of the reduced bundles must be assumed, since not every principal bundle admits a reduction\footnote{In the physical situation considered in this paper, such a reduction is indeed available.}. In contrast, for extensions no such assumption is needed, by virtue of \cref{th:existence of reduction/extension}.

Let us make a final comment: all results presented above \cref{sec: Fibre bundle}, formulated in terms of connections on principal bundles, can of course equivalently be expressed locally using the corresponding Yang-Mills fields. To be more precise, given a connection one-form $\omega$ on a principal $G$-bundle $\xi=(P,\pi,M)$ and an open subset $U\subseteq M$ of the base space, one can define the gauge field $\mathcal{A}\colon\Gamma(TU)\rightarrow \mathfrak{g}$ via the pullback of $\omega$ along a section $\sigma\in\Gamma(\xi)$ on $U$,
\begin{equation}
    \mathcal{A}\coloneqq\sigma^*\omega\;.
\end{equation}

In particular, \cref{th:existence of reduction/extension} implies that a reduction of $\xi$ to an arbitrary subgroup $H\subset G$ always exists locally. Indeed, the potential obstruction (i.e. the existence of a section on $\xi/H$) is global, and therefore disappear when one restricts to a sufficiently small patch. Thus, upon choosing local sections, any Yang–Mills field may be locally dressed according to a prescribed reduction. However, the discussion presented above captures the topological constraints that may obstruct a global reduction of the structure group, which cannot be seen from purely local considerations.

\section{Boundary description via formal loop groups}
\label{sec:Extensions and reductions in presence of Boundaries}
In the previous section we analysed in detail the situations in which the \textit{structure group} of a principal bundle is reduced by a normal subgroup or extended to a semidirect product (split extension). This is the case in general when we have a fixed base manifold $M$, and we want to restrict or extend the group under which we study gauge transformations.

In our particular case, we take the base manifold $M$ to be a space-time with a future null boundary, $\partial M = \Ib^+$. The introduction of a boundary has two main characteristics. First, we include the boundary of the manifold into the analysis, and therefore we need a transversal coordinate. Second, we restrict the gauge group by its fall-offs \emph{relative to the boundary}. 

In what follows, we will apply the results of the previous section by taking a formal expansion in terms of a transverse coordinate to $\partial M$, denoted as $\rf$, and consider the loop group generated by the formal loop algebra $L \mathfrak{g}$ as the new structure group, i.e. we construct the group generated by formal $\rf$-expansions with coefficients on $\mathfrak{g}$ at the boundary. There is a vast literature on loop groups and their representations; for a comprehensive introduction we refer the reader to Pressley and Segal’s book \cite{Pressley:1988qk} on loop groups, while Kac’s book \cite{Kac:1990gs} remains the standard reference for loop algebras.

This construction allows us to interpret the boundary $\partial M$ as the base manifold supporting a suitable principal bundle whose structure group has been enlarged to contain the fall-offs of the fields with respect to the bulk. 

With these motivations in place, we now proceed to the detailed construction. The starting point is a principal $G$-bundle $\xi\coloneqq(P,\pi,M)$ describing Yang-Mills theory on the asymptotically flat spacetime $M$. Its gauge group, defined in \eqref{def:gauge group G} and denoted here by $\mathcal{G}(\xi)$, has, in a local trivialisation, the Lie algebra given by $C^\infty$ maps from $U$ to $\mathfrak{g}$,
\begin{equation}
    \label{eq:Lie algebra of gauge transf on M}
    \mathrm{Lie}(\mathcal{G}(\xi))=\{\lambda\colon U\rightarrow\mathfrak{g}\}\;,
\end{equation}
where $U\subset M$ is an open subset and $\mathfrak{g}$ is the Lie algebra of $G$. Elements of \eqref{eq:Lie algebra of gauge transf on M} can be written as $\lambda=\lambda^a(x)T^a$, where $\{T^a\}_{a=1,\dots,\dim \mathfrak{g}}$ is a basis of $\mathfrak{g}$, and $x= \{x^\mu\}_{\mu=1,\dots,\dim M}$ are local coordinates on $U$.

To study gauge transformations at the boundary of $M$, a natural first attempt would be to consider the restriction of the bundle to $\partial M$, namely the principal $G$-bundle
\begin{equation}
    \xi_{\partial M}\coloneqq(P,\pi_{\partial M},\partial M) \;,
\end{equation}
where the projection is defined by $\pi_{\partial M}\coloneqq\pi|_{\mathrm{preim}_\pi(\partial M)}$. Observe that this construction simply restricts the base manifold of the bundle, while the structure group $G$ remains unchanged. Locally, the Lie algebra of gauge transformations for this restricted bundle is
\begin{equation}
    \label{eq:Lie algebra of gauge transf on partial M}
    \mathrm{Lie}(\mathcal{G}(\xi_{\partial M}))=\{\Lambda_0\colon U\cap\partial M\rightarrow\mathfrak{g}\}\;.
\end{equation}
Let $\mathbf{y}=\{y^i\}_{i=1,\dots,\dim M-1}$ be a set of local coordinates on $U\cap\partial M$. Then the above space consists of $\mathfrak{g}$-valued functions $\Lambda_0(\mathbf{y})$, which naturally identify with the gauge parameters defined in \eqref{Lambda0def}. Therefore, while this construction is the most immediate, it captures only $\Lambda_0$ and discards the large gauge parameter (the $\Lambda_+$ defined in \eqref{def: Lambda plus}, with all the logarithmic terms taken to vanish for simplicity), which are expected to be responsible for the subleading soft effects coming from the soft theorems. A more refined geometric procedure is required, which we now describe.

Large gauge transformation parameters, which diverge at the boundary, are in fact contributions originating from the bulk. Therefore, in order to include them in our analysis, we must consider the boundary $\partial M$ together with its past domain of dependence\footnote{The \textit{past domain of dependence} of $\partial M$ is the set of points $x \in M$ such that every future inextendible causal curve through $x$ intersects $\partial M$, see \cite{Wald:1984rg}.}, $D^-(\partial M) \subset M$.

In order to express the gauge parameter $\Lambda$ in an expansion analogous to that of \eqref{def: Lambda plus} (again, without the logarithmic terms), we need to introduce local coordinates on $D^-(\partial M)$, in a neighborhood containing $\partial M$. To this end, we recall that the Collar Neighborhood Theorem guarantees that there exists a neighborhood in $D^-(\partial M)$, which we call $\widehat{\partial M}$, containing the boundary $\partial M$ and such that $\widehat{\partial M}$ is diffeomorphic to $[0,\rho_0)\times\partial M$, with $\rho_0 >0$, where the first factor parametrizes the direction normal to the boundary.\footnote{We are only considering $\partial M$ as $\Ib^+$, where the Collar Neighborhood is well defined towards the ``inside'' of the spacetime manifold. Other signatures of $\partial M$, such as space- or time-like may require a different treatment.} This allows us to define a coordinate system $(\rho,\mathbf{y})$, where $\mathbf{y}$ are the local coordinates on $\partial M$ and $\rho\in[0,\rho_0)$ is the transverse coordinate, in which the large gauge parameters can be systematically expanded.

The correct geometric step is therefore not to restrict the bundle $\xi$ merely to the boundary $\partial M$, but rather to the region $\widehat{\partial M}$. We then define the principal $G$-bundle $\xi_{\widehat{\partial M}}\coloneqq(P,\pi_{\widehat{\partial M}},\widehat{\partial M})$, with projection $\pi_{\widehat{\partial M}}\coloneqq\pi|_{\mathrm{preim}_\pi(\widehat{\partial M})}\;$. The gauge Lie algebra for this bundle is given by
\begin{equation}
    \label{eq:Lie algebra of gauge transf on hat partial M}
    \mathrm{Lie}(\mathcal{G}(\xi_{\widehat{\partial M}}))=\{\Lambda\colon U\cap\widehat{\partial M}\rightarrow\mathfrak{g}\}\;.
\end{equation}
 In general, given a function $f$ on $U\cap\widehat{\partial M}$ with non-vanishing limit on $\partial M$, it can be expanded using a Laurent series as
\begin{equation}
    f(\rho,\mathbf{y})=\sum_{n\leq 0}\rho^nf^{(-n)}(\mathbf{y})\;.
\end{equation}
In these coordinates, the boundary is located at $\rho=0$. To make contact with the coordinates $(\mathfrak{r},\mathbf{y})$ used in \cref{sec:transformation_at_group_level}, we observe that $\rho$ is related to $\mathfrak{r}$ via
\begin{equation}
    \rho=1/\mathfrak{r} \;.
\end{equation}
This identification allows us to consistently translate expressions between the two coordinate systems. In our particular case, we have gauge parameters with finite and divergent components at the boundary, so they can be expressed as a Laurent series
\begin{equation}
    \label{eq:Laurent expansion of Lambda in rho}
    \Lambda(\rho,\mathbf{y})=\sum_{n\leq 0}\rho^n\Lambda^{(-n)}(\mathbf{y})=\sum_{n\geq 0}\mathfrak{r}^n\Lambda^{(n)}(\mathbf{y}).
\end{equation}
The key observation is that the above expression can be interpreted in two equivalent ways: the parameters $\Lambda$ with non-vanishing limit on the boundary can be regarded either as
\begin{enumerate}
    \renewcommand{\labelenumi}{\roman{enumi}.}
    \item $\mathfrak{g}$-valued maps on $\widehat{\partial M}$, or
    \item $L\mathfrak{g}$-valued maps on $\partial M$,
\end{enumerate}
where $L\mathfrak{g}$ is the formal loop algebra of $\mathfrak{g}$. Explicitly, in case ii., we think of the $\Lambda^{(n)}$ as objects defined on $\partial M$, extracted from $\Lambda$ via
\be 
\Lambda^{(n)} = \frac{1}{2\pi i}\text{Res}_{\rho = 0} \left( \rho^{n-1}\Lambda \right) = \frac{1}{2\pi i} \text{Res}_{\mathfrak{r}=\infty} \left( {\mathfrak{r}}^{-n-1} \Lambda \right) , \quad n\geq 0.
\ee 

Let us pause the discussion and briefly outline the concept of a formal loop algebra, providing the framework for properly describing the second interpretation above. Consider the coordinate $\rf$ transversal to the boundary $\partial M$. We can take formal expansions in the real parameter $\rf$ as follows: let $L\mathfrak{g}$ be the formal loop algebra

\be 
L\mathfrak{g} := \C [\rf ,\rf^{-1}] \otimes \mathfrak{g} ,
\ee
with Lie bracket
\be 
[p(\rf)\otimes a , q(\rf) \otimes b]_{L\mathfrak{g}} = p(\rf)q(\rf) \otimes [a,b]_{\mathfrak{g}}, \quad p(\rf),q(\rf)\in \C [\rf ,\rf^{-1}] ;\quad a,b \in \mathfrak{g},
\ee
where $[\cdot , \cdot]_{\mathfrak{g}}$ is the Lie bracket in $\mathfrak{g}$, $\otimes$ is the tensor product and $\C [\rf ,\rf^{-1}] $ is the algebra of Laurent polynomials in $\rf$. In the case of a periodic $\mathfrak{r}$, the algebra is the actual \emph{loop} algebra of maps form $S^1$ to $\mathfrak{g}$ (see \cite{Kac:1990gs} for details), hence the name \emph{formal} in our case.

We consider the exponentiation of this algebra, which we will call the \emph{formal loop group} $L G$ associated to $G$,
\be 
\text{exp}: L \mathfrak{g} \rightarrow L G, \quad  p(\rf) \otimes a \mapsto e^{i p(\rf) \otimes a}\;.
\ee

In particular, in the second viewpoint of the series in \eqref{eq:Laurent expansion of Lambda in rho}, the parameters $\Lambda$ naturally form the Lie algebra of the principal bundle with base manifold $\partial M$ and a structure group within the formal loop group $LG$.

We are interested in the subgroup in $LG$ generated by \emph{non-vanishing} transformations. Consider the following subalgebras\footnote{The reader may be more familiar with the usual triangular decomposition of a loop algebra, which can be defined as follows. Let $\mathfrak{g} = \mathfrak{n}_- \oplus \mathfrak{h} \oplus \mathfrak{n}_+ $ be the triangular decomposition of $\mathfrak{g}$. We can lift this decomposition into a triangular decomposition of $L\mathfrak{g}$ as follows
\begin{equation}
    \begin{aligned}
            \mathfrak{l}_- &:= (\rf^{-1} \C [\rf^{-1}] \otimes (\mathfrak{h} \oplus \mathfrak{n}_+)) \oplus \C[\rf^{-1}] \otimes \mathfrak{n}_-, \\
     \mathfrak{l}_+ &:= (\rf \C [\rf] \otimes (\mathfrak{h} \oplus \mathfrak{n}_-)) \oplus \C[\rf] \otimes \mathfrak{n}_+, \\
     \mathfrak{l}_0 &:= \mathfrak{h}.
    \end{aligned}
\end{equation}
In our decomposition \eqref{eq:our decomposition}, $\tilde{\mathfrak{l}}_0$ is no longer the Cartan subalgebra, but for our purposes this does not restrict the analysis.
}

\begin{equation}
    \begin{aligned}\label{eq:our decomposition}
            \tilde{\mathfrak{l}}_- &:= \rf^{-1} \C [\rf^{-1}] \otimes \mathfrak{g}, \\
    \tilde{\mathfrak{l}}_+ &:= \rf \C [\rf] \otimes \mathfrak{g}, \\
    \tilde{\mathfrak{l}}_0 &:= \C \otimes \mathfrak{g}.
    \end{aligned}
\end{equation}
We can then define the following subgroups in $LG$:
\begin{equation}
    \begin{aligned}
        N_- &:= e^{\tilde{\mathfrak{l}}_-}, \\
        N_+ &:= e^{\tilde{\mathfrak{l}}_+}, \\
        N_0 &:= e^{\tilde{\mathfrak{l}}_0}.
    \end{aligned}
\end{equation}
Note that by equation \eqref{eq:normality}, we have that $N_\pm \triangleleft \left\langle N_\pm , N_0 \right\rangle $, where $\left\langle , \right\rangle$ denote the generated group. This is a key observation, and allows us to use \cref{th:connection on split extension,th:connection on reduction to quotient group}.

The natural definition for the \emph{boundary} structure group is then the one generated by $N_+$ and $N_0$,
\be \label{def:G partial M}
G_{\partial M} := \left\langle N_+ , N_0 \right\rangle \subset LG.
\ee
We denote the new bundle associated to $G_{\partial M}$ by
\begin{equation}
    \xi_{\partial M}\coloneqq(P',\pi'_{\partial M},\partial M)\;.
\end{equation}
The gauge Lie algebra for this bundle is given by
\begin{equation}
    \label{eq:Lie algebra of gauge transf on Loop partial M}
    \mathrm{Lie}(\mathcal{G}(\xi_{\partial M}))=\{\Lambda\colon U\cap\partial M\rightarrow L\mathfrak{g}\}\;.
\end{equation}
Schematically, we have a commuting diagram,
\begin{equation}
    \begin{tikzcd}[column sep=3cm, row sep=huge]
        \xi \arrow[r, "M\rightarrow\partial M"] \arrow[d, "M\rightarrow \widehat{\partial M}"'] &  \xi_{\partial M} \arrow[d, "G\rightarrow G_{\partial M}"] \\
        \xi_{\widehat{\partial M}} \arrow[r, "\text{$\widehat{\partial M}\rightarrow\partial M$}", "\text{$G\rightarrow G_{\partial M}$}"'] & \xi_{\partial M}
    \end{tikzcd}
\end{equation}

In this construction, there is a trade-off between gauge transformations as smooth sections\footnote{Recall that a gauge transformation of a principal bundle $\xi$ can be viewed as a section of the associated bundle $\mathrm{Ad}(\xi)$, see observation at the end of \cref{sec:gauge transf}.} on $M$ on one side, and a formal expansion of smooth sections on $\partial M$ on the other.
\begin{figure}[h!]
    \centering
    \includegraphics[width=0.45\linewidth]{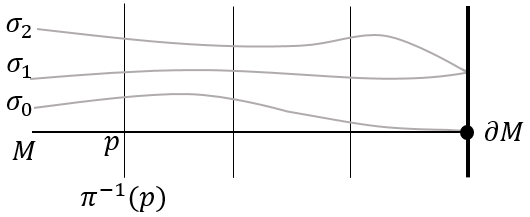}
    \caption{Schematic representation of the map from gauge group on $M$ and the gauge group on $\partial M$.}
    \label{fig_def_loop}
\end{figure}
In \autoref{fig_def_loop} we show three sections, $\sigma_0$, $\sigma_1$ and $\sigma_2$, which we are taking as $G$-valued smooth functions on $M$. Section $\sigma_0$ tends to the identity element $\mathbf{1}_G$ on the boundary. Therefore, by definition, its corresponding element on the gauge group associated to the boundary will have values on $N_-$. Then, we can say it does not survive the ``limiting process'' to the boundary. Sections $\sigma_1$ and $\sigma_2$ converge to the same value at $\partial M$. Nevertheless, their fall-offs are different, and therefore they are mapped to different elements on $G_{\partial M}$. Indeed, say 
\be 
\sigma_1(\mathbf{y}) = \Lambda_0 (\mathbf{y}) + r f_1(\mathbf{y}), \quad \sigma_2(\mathbf{y}) = \Lambda_0 (\mathbf{y}) + r f_2(\mathbf{y}),
\ee
for $f_1$ and $f_2$ different functions. Then 
\be 
\lim_{\mathfrak{r} \rightarrow +\infty} \partial_{\mathfrak{r}} \sigma_1 \rightarrow f_1(\mathbf{y}) \neq f_2(\mathbf{y}) = \lim_{ \mathfrak{r} \rightarrow +\infty} \partial_{\mathfrak{r}} \sigma_2,
\ee
and so their representation in the loop group is different.

With this picture in mind, we can interpret $N_0$ as the group corresponding to the residual symmetry at the boundary responsible for the finite large gauge symmetries of the theory. We can also interpret $N_+$ as the (bigger) group corresponding to the large gauge symmetries of order $O(r)$ and higher.

\subsection{Stueckelberg procedure}

In order to make the connection with the results stated in the previous section, within the framework developed above, below we explain how the extension/reduction for the new principal bundle $\xi^L_{\partial M}$ lead naturally to the Stueckelberg procedure shown in \autoref{sec:transformation_at_group_level}.

We define the gauge transformations at the boundary as follows.
\begin{definition}[gauge transformation at the boundary]
    Let $\xi=(P,\pi,M)$ be a principal bundle with structure group $G$. A gauge transformation at the boundary of $M$ is a gauge transformation of the restricted bundle $\xi^L_{\partial M}\coloneqq(P',\pi'_{\partial M},\partial M)$ with structure group $G_{\partial M}$ as defined in \eqref{def:G partial M}.
\end{definition}
Thus, gauge transformations at the boundary are vertical principal bundle automorphisms of $\xi_{\partial M}$. We reproduce the explicit definition here for convenience:
\be 
    \mathrm{Aut}^v_\mathrm{p.bdl} (\xi_{\partial M}) := \{ \Phi : P_{\partial M} \rightarrow P_{\partial M} \mid \Phi^* \pi_{\partial M} = \pi_{\partial M}, \, \, R^*_{g} \Phi = R_{g} \Phi \quad \forall g \in G_{\partial M}  \}\;.
\ee

Let $\mathcal{H}$, $\mathcal{N}$ and $\mathcal{G}$ be the gauge group corresponding to $N_0$, $N_+$ and $G_{\partial M}$, respectively. These coincide with the gauge groups already defined in \cref{sec:transformation_at_group_level}, whose elements are now viewed as formal expansions on $\mathfrak{r}$ with support on $\partial M$.

Now, we are in position to apply \cref{th:connection on split extension,th:connection on reduction to quotient group}.

\begin{itemize}
    \item \textbf{Two-step procedure}. We consider the normal chain
    \be 
        \{\mathbf{1}_{G_{\partial M}} \} \triangleleft N_+ \triangleleft G_{\partial M}.
    \ee
    
    Consider a gauge field $\mathcal{A}$ that transforms under $\mathcal{H}$ but is invariant under $\mathcal{N}$, as it is the situation of the gauge field in \cref{sec:two_step_stueck}. On the one hand, by the reduction procedure in \cref{th:connection on reduction to quotient group}, we can dress $\mathcal{A}$ via
    \be 
        \mathbfcal A = s_0^{-1} \mathcal{A} s_0 + s_0^{-1} \dd s_0 \; ,
    \ee
    where $s_0$ is a $N_0$-valued smooth function on $P_{\partial M}$ transforming as \eqref{s_0Transform}, such that $\mathbfcal A$ is now invariant under $G_{\partial M}$. 
    
    On the other hand, since $N_+$ is normal in $G_{\partial M}$, we can use \cref{th:connection on split extension} to dress the gauge field $\mathcal{A}$, 

    \be 
    \mathcal A \; \mapsto \;  \tilde{\mathcal{A}} = s^{-1}_+  \mathcal A  s_+ + s^{-1}_+  \dd(s_+),
    \ee
    where $s_+$ is a $G_{\partial M}$-valued smooth function incorporating the new symmetries, transforming as \eqref{res:trasformation of s plus}. The final transformation $\mathbfcal A \mapsto \tilde{\mathcal{A}}$ is 
    \be 
    \mathbfcal A \; \mapsto \;  \tilde{\mathcal{A}} = s^{-1}_+ s_0\mathbfcal A s^{-1}_0 s_+ + s^{-1}_+ s_0 \dd(s^{-1}_0 s_+).
    \ee 
    In this way, we start with a semi-invariant gauge field, $\mathcal{A}$, reduce it to be full invariant, $\mathbfcal A$, and extend the symmetries to the group $G_{\partial M}$ via a dressing to $\tilde{\mathcal{A}}$. 
    
    \item \textbf{Single-step extension}. Take the trivial normal subgroup, 
    \be 
        \{\mathbf{1}_{G_{\partial M}} \} \triangleleft G_{\partial M}.
    \ee
    
    We start with a gauge field $\mathbfcal A$ which is completely invariant under the gauge group associated to $G_{\partial M}$, and we want to construct a gauge field $\tilde{\mathcal{A}}$ that captures transformations under the full gauge group. Since the trivial subgroup $\{\mathbf{1}_{G_{\partial M}} \}$ is normal in $G_{\partial M}$, the extension in \cref{th:connection on split extension} gives us a straightforward dressing,
    \be 
    \mathbfcal A \; \mapsto \; \tilde{\mathcal{A}} = s^{-1}  \mathbfcal A s + s^{-1} \dd(s),
    \ee 
    where $s$ is a $G_{\partial M}$-valued smooth function on $P_{\partial M}$ transforming as \eqref{full s transform}, and contains the full set of symmetries given by $G_{\partial M}$, thus the name Stueckelberg field given to it in \cref{subsec:review Stueck}.
\end{itemize}

As we mentioned in \cref{sec:two_step_stueck}, the two approaches can be related by identifying $s^{-1}_0 s_+$ in the first procedure with $s$ in the second procedure. The next diagram helps to visualize the two Stueckelberg procedures:
\begin{equation}
    \begin{tikzcd}[column sep=large, row sep=large]
        \mathcal{A} \arrow[r, "s_+"] \arrow[d,shift left=1mm, swap, "s_0^{-1} \;\; "] &
        \tilde{\mathcal{A}} \\
        \mathbfcal A \arrow[u, shift left=1mm, "\;\; 
        s_0 "']\arrow[ur, swap, "s"] &
    \end{tikzcd} 
\end{equation}

\section{Conclusions}
\label{sec: Conclusions}

In this article, we have significantly extended our previous proposal for extending the phase space such that it incorporates symmetries responsible for sub$^n$-leading soft theorems in a mathematically robust framework, via the so-called Stueckelberg procedure. We have achieved this in three main directions. Firstly, we constructed an explicit boundary action for the Stueckelberg fields, from which the charges associated with the subleading soft effects can be computed directly. Secondly, we performed the explicit renormalisation of the symplectic potential appearing in the calculation of the charges, both in the radial and in the $u$-direction. Finally we provide a formal derivation of the Stueckelberg procedure via fibre bundles, which leads us to a surprisingly clear picture of the Lie algebra at the boundary as a loop algebra, and its corresponding exponentiation as the structure group.

A natural direction for future work is to extend our construction to the gravitational setting. The Stueckelberg procedure is well suited to incorporate spacetime symmetries \footnote{See e.g.\cite{Henneaux:1989zc,Nagy:2019ywi,Bansal:2020krz,Bansal:2024pgg}.}, and edge modes have already been explored in the context of gravity \cite{Donnelly:2016auv, Speranza:2022lxr, Freidel:2021cjp, Donnelly:2022kfs}, and the dressing field method has been studied for gravitational theories \cite{Francois:2021jrk,Francois:2024rfh,Francois:2024rdm,Francois:2024vlr,Francois:2025jro,Berghofer:2025ius}. It would also be worthwhile to investigate whether the fibre bundle framework can provide a nice geometrical picture in this case.

The inclusion of loop effects \cite{He:2014bga, Bianchi:2014gla,Sahoo:2018lxl, Pasterski:2022djr, Donnay:2022hkf,Agrawal:2023zea,Choi:2024ygx,Campiglia:2019wxe,AtulBhatkar:2019vcb,Choi:2024mac,Choi:2024ajz} is known to generate logarithmic corrections to soft theorems and their associated symmetries. Our framework is, in principle, well suited to accommodate such modifications (see e.g \eqref{A log expansion}, \eqref{def: Lambda plus} and \eqref{def:dressing field s plus}), and a natural continuation of this work is to apply it in order to systematically incorporate these contributions. In particular, we intend to extend our boundary Lagrangian so as to capture loop-level effects and to construct the corresponding renormalised charges. We are also interested in recent developments on the relationship between soft and edge modes and the dressing by Wilson lines \cite{Carrozza:2021gju, Carrozza:2022xut, Araujo-Regado:2024dpr}. We leave the study of the relation of dressings by Wilson lines with the Stueckelberg procedure presented in this article for future work. This program can be pursued both in gauge theory and in gravity, and it has the potential to shed light on how infrared structures are preserved or modified once quantum corrections are taken into account.

Finally, an important avenue for future investigation is to deepen the understanding of the interplay with higher-spin algebras—specifically, the $S$-algebra in Yang–Mills theory and the $w_{1+\infty}$ loop algebra in gravity\footnote{See \cite{Pope:1989ew,Bakas:1989xu,Fairlie:1990wv,Pope:1991ig,Pope:1991zka,Strominger:2021mtt,Guevara:2021abz,Himwich:2021dau,Jiang:2021ovh,Boyer:1985aj,Park:1989fz,Park:1989vq,Adamo:2021lrv,Monteiro:2022lwm,Bu:2022iak,Bittleston:2023bzp,Bittleston:2024rqe,Taylor:2023ajd,Kmec:2024nmu,Nagy:2024dme,Lipstein:2023pih,Geiller:2024bgf,Freidel:2021ytz,CarrilloGonzalez:2024sto,Agrawal:2024sju,Kmec:2025ftx,Cresto:2025bfo,Baulieu:2024rfp}}. Progress in this direction would bring us closer to clarifying the profound connections between higher-spin symmetries and asymptotic symmetry structures. Encouragingly, these algebras naturally emerge in the self-dual sectors of the theories, where we have already begun to explore asymptotic symmetries and their associated Stueckelberg dressings \cite{Campiglia:2021srh,Diaz-Jaramillo:2025gxw,Nagy:2022xxs}.

\acknowledgments
 
We thank Sangmin Choi, Nicolas Cresto, Laurent Freidel, Marc Geiller, Sucheta Majumdar, Andrea Puhm and Celine Zwikel for useful discussions. G.P. is funded by STFC Doctoral Studentship 2023. S.N. is supported in part by STFC consolidated grant T000708. J.P. was partially supported by Fondo Clemente Estable Project FCE\_1\_2023\_1\_175902, CSIC grant C013-347, and funded by Postdoctoral Fellowship at Concordia University, Montreal. Research at Perimeter Institute is supported in part by the Government of Canada through the Department of Innovation, Science and Economic Development Canada and by the Province of Ontario through the Ministry of Colleges and Universities. This work was supported by the Simons Collaboration on Celestial Holography.

\appendix
\section{Some derivations}
\label{App:some derivations}
Consider a set of fields $A_1(\mathfrak{r},\mathbf{y}),\dots,A_N(\mathfrak{r},\mathbf{y})$ which admit an expansion in all integer powers of $\mathfrak{r}$
\be \label{App:expansion}
 A_i(\mathfrak{r},\mathbf{y})=\sum_{n\in\mathbb{Z}}\mathfrak{r}^nA_i^{(n)}(\textbf{y})\; ,\quad i=1,\dots,N
\ee 
Consider a renormalised product of the above, in the sense that we only keep the $\mathfrak{r}^0$ coefficient of the product (we omit the $\mathbf{y}$ dependence for simplicity):
\be 
\left[A_1...A_N \right]^{(0)}= \sum_{n_1+...+n_N=0}A_1^{(n_1)}...A_N^{(n_N)}
\ee
Then we see that variation commutes with renormalisation:
\be 
\begin{aligned}
\delta\left[A_1...A_N \right]^{(0)}&=\delta \sum_{n_1+...+n_N=0}A_1^{(n_1)}...A_N^{(n_N)}\\
&=\sum_{n_1+...+n_N=0}\left(\delta A_1^{(n_1)}...A_N^{(n_N)}+...+A_1^{(n_1)}...\delta A_N^{(n_N)}\right)\\
&=\left[\delta A_1...A_N+...+A_1...\delta A_N \right]^{(0)}
\end{aligned}
\ee 
upon using the following definition for the variation of the coefficients of $\mathfrak{r}$ in \eqref{App:expansion}:
\be 
\delta A_1^{(n_1)}\equiv [\delta A_1]^{(n_1)}
\ee 
We can now write the variation of the boundary part of the action in \eqref{bound Lag} as 
\be 
\begin{aligned}
\delta S_{\partial M}&=\int_{\partial M}\tr[\mathbf j\wedge(\delta \mathbfcal A+\dd(\delta s)s^{-1}-\dd s(s^{-1}\delta ss^{-1}))]^{(0)} \\
    &=\int_{\partial M}\tr[\mathbf j\wedge\delta \mathbfcal A -\dd \mathbf j\delta ss^{-1}-\mathbf j\wedge\dd ss^{-1}\delta ss^{-1}+\mathbf j\wedge \delta ss^{-1}\dd ss^{-1}\\
    &\quad\quad\quad\quad\ +\dd(\mathbf j \delta ss^{-1})]^{(0)}\\
    &=\int_{\partial M}\tr[\mathbf j\wedge\delta \mathbfcal A -\left(s\dd\left(s^{-1} \mathbf j s\right)s^{-1}\right)\delta ss^{-1} +\dd(\mathbf j \delta ss^{-1})]^{(0)} \; .
\end{aligned}
\ee

\section{Lie group actions on manifolds}
\label{App:{Lie group actions on manifolds}}
In what follows we recall some basic notions about Lie groups and their relation with manifolds, which will play an important role in the discussion in \autoref{sec: Fibre bundle}. Throughout, by \enquote{manifold} we always mean a smooth manifold.

Let $G$ be a Lie group and $H\subseteq G$ a Lie subgroup. The \textit{left coset} of $H$ in $G$ with representative $g\in G$ is defined as
\begin{equation}
    gH\coloneqq[g]_H=\{gh \mid h\in H\} \;,
\end{equation}
where $[g]_H$ denotes the equivalence class of $g$ with respect to the relation $g\sim_Hg'$ if and only if $g'=gh$ for some $h\in H$. The collection of all left cosets forms the \textit{quotient space}
\begin{equation}
    G/H\coloneqq\{gH\mid g\in G\} \;,
\end{equation}
with the natural projection
\begin{equation}
    \rho_H\colon G\rightarrow G/H\,, \quad g\mapsto [g]_H
\end{equation}
called the \textit{quotient map}. A further important concept is that of a normal subgroup. A subgroup $N\subseteq G$ is said to be \textit{normal}, written $N\triangleleft G$, if it is stable under conjugation in $G$, that is
\begin{equation}
    \mathrm{Ad}_g(n)\in N \quad \forall g\in G,n\in N\;,
\end{equation}
where $\mathrm{Ad}_g\in\mathrm{Aut}(G)$ is the map $\mathrm{Ad}_g(g')\coloneqq gg'g^{-1}$. The collection of all left cosets of $N$ in $G$ (i.e. the quotient space $G/N$) forms a group, called the \textit{quotient group}.

The next key notion is that of group actions on manifolds, which formalize how a Lie group can act smoothly on the points of a manifold. This concept will be crucial in the present section, as it provides the framework for discussing (principal) fibre bundles and, eventually, the transformation properties of the Stueckelberg fields.
\begin{definition}[group action]
    The left $G$-action of a Lie group $G$ on a manifold $M$ is a map
    \begin{equation}
        \la\colon G\times M\rightarrow M\,, \quad (g,m)\mapsto g\la m
    \end{equation}
    such that for all $g_1,g_2 \in G$ and $m\in M$
    \begin{align}
        \label{def: left action prop}
        e\la m=m\,, \quad (g_1g_2)\la m=g_1\la g_2\la m \;,
    \end{align}
    where $e\in G$ is the identity element. $M$ equipped with such a map is called a left $G$-manifold.
\end{definition}
An analogous definition can be given for a right action $\ra\colon M\times G\rightarrow M$ and a right $G$-manifold\footnote{In this case, the defining properties in \eqref{def: left action prop} become, respectively, $m\ra e=m$ and $m\ra(g_1g_2)=m\ra g_1\ra g_2$.}. Given a fixed element $g\in G$, one can define the maps $L_g,R_g\colon M\rightarrow M$ as
\begin{equation}\label{def:action maps}
    L_g(m)\coloneqq g\la m\,, \quad R_g(m)\coloneqq m\ra g
\end{equation}
which are referred to as the left and right action maps associated with $g$, respectively, and will be very useful later on to simplify the notation. When using the above maps throughout the paper we often adopt the convention of omitting parentheses; for example, $L_g(m)$ is simply written as $L_gm$.

Whenever two manifolds endowed with group actions are involved, it is natural to ask for maps between them that are compatible with the respective actions. Such maps are called equivariant maps. Intuitively, an equivariant map is one that commutes with the action of the groups.
\begin{definition}[equivariant map]
    Let $G,H$ be Lie groups and $\rho\colon G\rightarrow H$ a Lie group homomorphism. Let $M,N$ be manifolds with left\footnote{An analogous definition holds for right actions.} $G$-action $\la$ and left $H$-action $\blla$, respectively. A map $f\colon M\rightarrow N$ is called $\rho$-equivariant if
    \begin{equation}
        f(g\la m)=\rho(g)\blla f(m) \quad \forall g\in G, m\in M\;.
    \end{equation}
\end{definition}
To analyse the geometry of group actions in more detail, it is useful to examine the sets of points that a group can move a given point to, as well as the subgroups that leave points fixed. These lead to the fundamental notions of orbits and stabilizers.
\begin{definition}[orbit, stabiliser]
    Let the group $G$ act on the manifold $M$ with a left action $\la$. The orbit of a point $m\in M$ is defined to be the set
    \begin{equation}
        \mathscr{O}_m\coloneqq\{m'\in M\mid \exists g\in G \; \colon \ m'=g\la m\} \;,
    \end{equation}
    while the stabiliser (or isotopy group) of $m$ is the subgroup
    \begin{equation}
        S_m\coloneqq\{g\in G\mid g\la m=m\}\;.
    \end{equation}
\end{definition}
Orbits partition the manifold into disjoint sets, and the collection of these orbits forms the \textit{quotient space}
\begin{equation}
    M/G\coloneqq\{\mathscr{O}_m\subseteq M\mid m\in M\}\;,
\end{equation}
that captures the manifold’s structure up to the group action. If we regard the orbit $\mathscr{O}_m$ as a point in $M/G$, we might also denote it by $[m]_G$. This notation is natural, as the orbits of $G$ are the equivalence classes of the relation $m\sim_G m'$ if and only if $m'\in\mathscr{O}_m$, where $m,m'\in M$. For this reason, the quotient space can be equivalently written as $M/\sim_G$.

Certain types of group actions are widely studied in differential geometry and also play a role in our discussion of fibre bundles.
\begin{definition}[free action]
    A $G$-action on a manifold $M$ is called free if $S_m=\{e\}$ for all $m\in M$, where $e$ is the identity element of $G$.
\end{definition}

To conclude this short review, let us recall the definition of pullback. Let $V,W,Z$ be vector spaces and $M,N$ be smooth manifolds. The \textit{pullback of a map} $g\colon W\rightarrow Z$ by a function $f\colon V\rightarrow W$ is $f^*g\coloneqq g\circ f$. The \textit{pullback of a form} $\omega\in T^*_{\phi(p)}N$ by a function $\phi\colon M\rightarrow N$ at a point $p\in M$ is defined by $(\phi^*)_p(\omega)(X)\coloneqq\omega((\phi_*)_p(X))$, where $(\phi_*)_p$ is the pushforward of $\phi$ at $p$ and $X\in T_p M$.



\bibliography{Ref_all_r_SD}
\bibliographystyle{utphys}

\end{document}